\documentclass[times,review,10pt]{elsarticle}

\usepackage{url}

\usepackage{mathtools}
\usepackage{amsmath}
\usepackage{amssymb}
\usepackage{amsthm}
\usepackage{booktabs}
\usepackage{multirow}
\usepackage{longtable}
\usepackage{hyperref}
\usepackage{caption}
\newtheorem{theorem}{Theorem}[section]
\newtheorem{corollary}{Corollary}[theorem]

\newtheorem{proposition}[theorem]{Proposition}
\usepackage{xcolor}
\newcommand{\defeq}{\vcentcolon=}

\usepackage{algorithm}
\usepackage[noend]{algpseudocode}
\algnewcommand\algorithmicforeach{\textbf{for each}}
\algdef{SE}[FOR]{ForEach}{EndForEach}[1]{\algorithmicforeach\ #1\ \algorithmicdo}{\algorithmicend\ \algorithmicforeach}%

\journal{Fuzzy Sets and Systems}

\begin{document}

\begin{frontmatter}

%% Title, authors and addresses

%% use the tnoteref command within \title for footnotes;
%% use the tnotetext command for theassociated footnote;
%% use the fnref command within \author or \affiliation for footnotes;
%% use the fntext command for theassociated footnote;
%% use the corref command within \author for corresponding author footnotes;
%% use the cortext command for theassociated footnote;
%% use the ead command for the email address,
%% and the form \ead[url] for the home page:
%% \title{Title\tnoteref{label1}}
%% \tnotetext[label1]{}
%% \author{Name\corref{cor1}\fnref{label2}}
%% \ead{email address}
%% \ead[url]{home page}
%% \fntext[label2]{}
%% \cortext[cor1]{}
%% \affiliation{organization={},
%%            addressline={}, 
%%            city={},
%%            postcode={}, 
%%            state={},
%%            country={}}
%% \fntext[label3]{}

\title{Spectrally Tuned Bandwidth Selection for Kernel Fuzzy Relational Clustering} %% Article title

%% use optional labels to link authors explicitly to addresses:
%% \author[label1,label2]{}
%% \affiliation[label1]{organization={},
%%             addressline={},
%%             city={},
%%             postcode={},
%%             state={},
%%             country={}}
%%
%% \affiliation[label2]{organization={},
%%             addressline={},
%%             city={},
%%             postcode={},
%%             state={},
%%             country={}}

\author[label1]{Efthymios Costa} %% Author name

%% Author affiliation
\affiliation[label1]{organization={Imperial College London},%Department and Organization
            addressline={Exhibition Road}, 
            city={London},
            postcode={SW7 2AZ}, 
            %state={},
            country={United Kingdom}}

\author[label2]{John R.J. Thompson} %% Author name

%% Author affiliation
\affiliation[label2]{organization={The University of British Columbia},%Department and Organization
            addressline={3333 University Way}, 
            city={Kelowna},
            postcode={V1V 1V7}, 
            state={British Columbia},
            country={Canada}}

%% Abstract
\begin{abstract}
Fuzzy clustering is used to identify overlapping geometric cluster structures through partial memberships. However, classical methods are limited by the assumption of equal variable importance and by sensitivity to the fuzzifier parameter. These limitations may yield equal cluster membership probabilities, which we refer to as the uniform solution. To address these issues, we propose Kernel Fuzzy Relational Clustering (KFRC) equipped with a bandwidth selection algorithm tuned via the spectral properties of the induced kernel Gram matrix. The KFRC framework implicitly performs unsupervised kernel metric learning by controlling the geometric embedding of the data through adjustable bandwidth parameters. We conduct a formal stability analysis to identify the exact theoretical conditions under which relational clustering collapses, thereby ensuring the stable performance of KFRC. We find that our two-stage bandwidth selection procedure adapts to the data structure while actively avoiding the uniform solution. Furthermore, this theoretical analysis leads to the proposal of a novel fuzzifier function that presents distinct advantages over the power fuzzifier function. We conduct experiments on several synthetic and publicly available data sets to demonstrate that the proposed framework recovers complex structures that traditional methods fail to resolve, while avoiding both uniform and near-crisp solutions.
\end{abstract}

%%Graphical abstract
%\begin{graphicalabstract}
%\includegraphics{grabs}
%\end{graphicalabstract}

%%Research highlights
% \begin{highlights}
% \item Develop a novel kernel-based fuzzy relational clustering algorithm.
% \item Design bandwidth selection to prevent collapse and ensure cluster quality.
% \item Define uniform collapse in fuzzy clustering mathematically.
% \item Conduct a sensitivity analysis to investigate the conditions for uniform collapse under mild assumptions and in full generality.
% \item Identify the conditions under which fuzzy relational clustering collapses.
% \item Introduce a fuzzifier function avoiding uniform collapse in many-cluster settings.
% \end{highlights}

%% Keywords
\begin{keyword}
Fuzzy clustering \sep relational data \sep kernel functions \sep bandwidth selection \sep fuzzifier selection
%% keywords here, in the form: keyword \sep keyword

%% PACS codes here, in the form: \PACS code \sep code

%% MSC codes here, in the form: \MSC code \sep code
%% or \MSC[2008] code \sep code (2000 is the default)

\end{keyword}

\end{frontmatter}

%% Add \usepackage{lineno} before \begin{document} and uncomment 
%% following line to enable line numbers
%% \linenumbers

\section{Introduction}
Clustering seeks to group $n$ objects $\mathbf{X}=\{\mathbf{x}_1,\ldots,\mathbf{x}_n\}$ measured on $p$ variables (that is, $\mathbf{x}_j = (x_{j1}, \ldots, x_{jp})^\top$) into groups or clusters of similar points, while simultaneously separating clusters that are dissimilar. However, sometimes data is not structured as objects-by-variables, but rather as relational data, which is an $n \times n$ matrix $\mathbf{R}$ of pairwise distances or dissimilarities $R_{j,k}$ between the $j$th and the $k$th objects. For instance, relational data is common in bioinformatics, when observations are alignment scores between DNA sequences of varying lengths \cite{maji2008rough}, and in social network analysis, where nodes may exhibit partial membership in multiple overlapping communities represented through pairwise relations \cite{nepusz2008fuzzy}.

The focus of the research here is clustering relational data. Clustering algorithms can be categorised into two groups based on their output: hard or {\it crisp} clustering, where each data point is assigned to exactly one cluster, and soft or {\it fuzzy} clustering, where objects may belong to more than one cluster. In fact, hard clustering is a special case of fuzzy clustering. 

In this paper, we focus on Fuzzy Relational Clustering (FRC). Several approaches for FRC have been proposed. One of the first FRC algorithms to appear was Relational Fuzzy $c$-Means (RFCM) \cite{hathaway1989relational}, which assumes a Euclidean structure for the input distance matrix. This was later extended to Non-Euclidean Relational Fuzzy $c$-Means (NERFCM) \cite{hathaway1994nerf}, which applies a mathematical transformation to convert arbitrary dissimilarities into Euclidean distances so that Fuzzy $c$-Means (FCM) \cite{Dunn01011973,bezdek1984fcm} can be applied. However, because these approaches require Euclidean geometry, they can distort the original dissimilarity structure. An alternative approach is the FANNY algorithm \cite{kaufman2009finding}. FANNY operates directly on the raw dissimilarity matrix by minimising a relational double-sum objective. This allows FANNY to handle any arbitrary dissimilarity function without modifying the original data or computing explicit cluster centres.

The degree of fuzziness in FRC algorithms is typically controlled by a fuzzifier function, which is parametrised by a constant denoted by $m$. Bezdek et al. \cite{bezdek1984fcm} empirically observed that $1.5 \leq m \leq 3$ generally yields good results for the FCM algorithm. This heuristic selection of $m$ was later refined to the interval $[1.5, 2.5]$ \cite{pal1995cluster}. While numerous studies have attempted to establish guidelines for optimal fuzzifier selection (summarised in \cite{gupta2019fuzzy}), the majority of these recommendations remain heavily reliant on heuristic arguments. To the best of our knowledge, the first rigorous theoretical link between the fuzzifier value and the underlying data geometry was established by Yu et al. \cite{yu2004analysis}. They demonstrated that FCM yields meaningful fuzzy partitions only when $m$ is constrained below a specific theoretical bound. However, their derivation inherently assumes explicit Euclidean coordinates and is limited to the classical, object-based FCM framework.

In this paper, we extend the work of Yu et al. to FRC algorithms that use distances defined in a Hilbert space. This is achieved by using kernel functions to define dissimilarities between objects. Our first contribution is a theoretical analysis of the dependence of the stability of generalised FRC algorithms to the fuzzifier function and its parameter. This generalises the result of Yu et al. and yields an upper bound for the fuzzifier parameter that prevents the algorithm from collapsing to the uniform solution. Based on this analysis, our second contribution is a bandwidth selection algorithm that allows arbitrarily large values of the fuzzifier parameter. Our third contribution is a novel fuzzifier function that presents theoretical advantages over the commonly used power fuzzifier.

The remainder of the paper is organised as follows: Section~\ref{sec:frc} formulates the general FRC objective, introduces Kernel FRC, and defines the uniform collapse phenomenon. Section~\ref{sec:stability} presents a stability analysis to derive the conditions under which collapse is avoided. Based on these conditions, we introduce a novel fuzzifier function. Section~\ref{sec:bw_select_alg} details the proposed bandwidth selection algorithm, and Section~\ref{sec:simulations} evaluates its performance on both synthetic and real-world data sets. We finally make some concluding remarks in Section~\ref{sec:conclusion}.

\section{Fuzzy Relational Clustering}\label{sec:frc}

In the most general setting, the problem of FRC for clustering a data set $\mathbf{X} \in \mathbb{R}^{n\times p}$ with dissimilarities $\mathbf{R}\in\mathbb{R}_{\ge 0}^{n\times n}$ into $c > 1$ clusters is formulated as the minimisation of an objective given by
\begin{equation}\label{eq:frc_general}
    J_\text{FRC}(\mathbf{U}) = \sum\limits_{i=1}^c \frac{\sum\limits_{j=1}^n \sum\limits_{k=1}^n t_m(u_{ij}) t_m(u_{ik}) R_{jk}}{2 \sum\limits_{l=1}^n t_m(u_{il})},
\end{equation}
where $u_{ij} \geq 0 \ \forall i,j$ subject to the constraint $\sum_{i=1}^cu_{ij} = 1 \ \forall j$, and $t_m:[0,1] \rightarrow [0,1]$ is the fuzzifier function. The matrix $\mathbf{U}$ is an element of the space of fuzzy partitions, given by
\begin{equation*}
    M_{fcn} \defeq \Big\{ \mathbf{U} \in \mathbb{R}^{c \times n} \ \Big| \  u_{ij} \in [0, 1] \ \forall i, j \wedge \sum_{i=1}^c u_{ij} = 1 \quad \forall j \in \{1, \dots, n\} \Big\}.
\end{equation*}
In a relational context, the elements of $\mathbf{R}$ are defined as $R_{jk} = d(\mathbf{x}_j, \mathbf{x}_k)$, where $d: A \times A \rightarrow \mathbb{R}_{\ge 0}$ is a dissimilarity function over a set $A$, for which the following three properties hold
\begin{align*}
d(x,y) \geq 0 \ \forall x,y \in A \quad & \text{(Non-negativity)}\\
d(x,y) = d(y, x) \ \forall x,y \in A \quad  & \text{(Symmetry)}\\
d(x,y) = 0 \iff x = y \quad & \text{(Identity of indiscernibles)}
\end{align*}
Distance or metric functions are a subset of the broader class of dissimilarities, which additionally require the triangle inequality to hold. 

Operating alongside these relational affinities $R_{jk}$, the objective function in Expression \eqref{eq:frc_general} employs a fuzzifier function $t_m$ to weigh membership assignments. This function is parametrised by the fuzzifier $m \in [m^\text{min}, m^\text{max}]$ which controls the fuzziness of the solution and whose minimum $m^\text{min}$ and maximum values $m^\text{max}$ depend on the definition of $t_m$. Larger values of $m$ yield a fuzzier partition, and $m=m^\text{min}$ produces a hard partition. Based on \cite{klawonn2003fuzzy}, fuzzifier functions must satisfy the following three properties
\begin{align*}
    t_m(0) = 0, \ t_m(1) = 1 \ \forall m & \text{ (Boundary guarding)}\\ 
    t_m(u) \in \mathcal{C}^1(0,1), \  t'_m(u) > 0 \ \forall u \in (0,1)  & \text{ (Monotonicity)}\\
    t_m(u) \in \mathcal{C}^2(0,1), \  t''_m(u) >0 \ \forall u \in (0, 1) & \text{ (Convexity)}
\end{align*}
An additional desirable (but not necessary) property for fuzzifier functions is origin sparsity, which requires $\lim_{u \rightarrow 0^+}t'_m(u) > 0$ for any finite value of $m$. Origin sparsity, also known as the `high contrast' property \cite{rousseeuw1995fuzzy}, allows for hard partitions to be obtained when clusters are well-separated, or, more generally, for hard and soft partitions to coexist even for $m > m^\text{min}$. 

While the formulation of the FRC objective in Expression \eqref{eq:frc_general} allows for the use of any arbitrary dissimilarity function, it retains a theoretical connection to object-based algorithms like FCM, which seeks to minimise the generalised least-squares functional $J_\text{FCM}$:
\begin{equation}\label{eq:fcm_objective}
    J_\text{FCM}(\mathbf{U}, \mathbf{V}) = \sum\limits_{j=1}^n\sum\limits_{i=1}^c u_{ij}^m \| \mathbf{x}_j - \mathbf{v}_i\|^2,
\end{equation}
where $\mathbf{V} \in \mathbb{R}^{c \times p}$ is the matrix of cluster centres $\mathbf{v}_i$. Expression \eqref{eq:fcm_objective} can be seen as a special case of $J_\text{FRC}$ when the dissimilarity function corresponds to a squared Euclidean distance $d(\mathbf{x}_j, \mathbf{x}_k) = \|\phi(\mathbf{x}_j) - \phi(\mathbf{x}_k)\|_{\mathcal{H}}^2$ for a feature map $\phi : \mathcal{X} \rightarrow \mathcal{H}$, where $\mathcal{H}$ is a Hilbert space. In this formulation, originally shown by Tucker \cite{tucker1987counterexamples}, the algorithm implicitly minimises the squared distances between the mapped observations and a set of cluster centroids residing in $\mathcal{H}$. This formulation shows that the relational objective becomes mathematically equivalent to an object-by-variable objective. This also explains why the object-based fuzzy clustering objectives, such as Expression \eqref{eq:fcm_objective}, are formulated as an optimisation problem over the fuzzy memberships $\mathbf{U}$ and the cluster centres $\mathbf{V}$; the latter is implicit in the formulation of general FRC objective functions.

\subsection{Kernel Fuzzy Relational Clustering (KFRC)}\label{subsec:kfrc}

The equivalence between FRC and object-based fuzzy clustering algorithms provides a means to explicitly control the geometric embedding of relational data through kernel functions. Since the equivalence requires the dissimilarities to represent squared Euclidean distances in an inner-product space, we can mathematically guarantee this holds by deriving the relational matrix from a symmetric, positive semi-definite (PSD) kernel function. The squared Euclidean distance between any two mapped observations in $\mathcal{H}$ can be computed entirely through the kernel trick \cite{aizerman1964theoretical} by
\begin{equation*}
    d_k(\mathbf{x}_j, \mathbf{x}_k) = \|\phi(\mathbf{x}_j) - \phi(\mathbf{x}_k)\|_{\mathcal{H}}^2 = k(\mathbf{x}_j, \mathbf{x}_j) + k(\mathbf{x}_k, \mathbf{x}_k) - 2k(\mathbf{x}_j, \mathbf{x}_k),
\end{equation*}
where $k: \mathcal{X} \times \mathcal{X} \rightarrow \mathbb{R}$ is a kernel function such that $k(\mathbf{x}_j, \mathbf{x}_k)$ is an inner product of $\phi(\mathbf{x}_j)$ and $\phi(\mathbf{x}_k)$ by $\langle \phi(\mathbf{x}_j), \phi(\mathbf{x}_k) \rangle_{\mathcal{H}} = k(\mathbf{x}_j, \mathbf{x}_k)$. Kernels are symmetric PSD if and only if $k(\mathbf{x}_j, \mathbf{x}_k)  = k(\mathbf{x}_k, \mathbf{x}_j) \ \forall \mathbf{x}_j, \mathbf{x}_k \in \mathcal{X}$ and $\sum_{j=1}^n \sum_{k=1}^n \alpha_j\alpha_k k(\mathbf{x}_j, \mathbf{x}_k) \geq 0$ for all $\mathbf{x}_{j}, \mathbf{x}_k \in \mathcal{X}$ and for all $\alpha_1, \ldots, \alpha_n \in \mathbb{R}$. The rationale behind the use of symmetric PSD kernel functions is that they define a unique Reproducing Kernel Hilbert Space (RKHS) $\mathcal{H}$ and an associated feature map $\phi: \mathcal{X} \rightarrow \mathcal{H}$. This mapping ensures that the inner product between any two mapped points within $\mathcal{H}$ can be evaluated exactly by the kernel function in the original input space. 

By using the kernel trick, we can control the geometry of the induced space. The usage of kernel functions enables clustering of complex data structures and provides a straightforward way to induce feature weighting via bandwidth selection. Substituting this kernel-induced distance $d_k$ into $J_\text{FRC}$, we obtain the Kernel Fuzzy Relational Clustering (KFRC) objective. The aggregation of multiple kernels for similarity quantification is a common approach in the clustering literature, typically achieved through either sums \cite{ghashti2025mixed} or products \cite{costa_deterministic_2025} of kernels. Both are mathematically valid, as the class of symmetric PSD kernels is closed under addition and multiplication. Geometrically, however, they represent fundamentally different embeddings: the sum kernel corresponds to a concatenation of the underlying feature maps, whereas the product kernel induces a tensor product space. Despite these distinct geometric implications, both methods are used in practice and our experiments in Section \ref{sec:simulations} do not reveal any consistent empirical advantage of one combination method over the other.

\subsection{Uniform collapse in fuzzy clustering}\label{subsec:unif_solution}

A fundamental challenge in clustering, and more impactfully in fuzzy clustering, lies in guiding the algorithm from a highly non-informative initialisation to a structured partition. To bypass this difficulty, many software implementations use multiple random initialisations of $\mathbf{U}$ and retain the solution that minimises the fuzzy clustering objective. However, this approach introduces bias as these random initialisations typically favour the allocation of observations to specific clusters.

A more rigorous approach is to start from a near-uniform state, assume almost no prior knowledge of the fuzzy membership matrix $\mathbf{U}$, and let the structure emerge. However, this carries a major risk: the algorithm may fail to find any meaningful clusters and instead collapse into a state where every observation is assigned an equal probability of belonging to every cluster. This phenomenon is referred to as {\it uniform collapse}, since the clustering ``collapses" to the undesirable uniform solution $(\mathbf{U}^*, \mathbf{V}^*)$, where $\mathbf{U}^* = \mathbf{1}_c \mathbf{1}_n^\top/c \in \mathbb{R}^{c \times n}$, $\mathbf{1}_d$ are uniform membership probabilities, and $\mathbf{V}^* = \bar{\phi}(\mathbf{x})\mathbf{1}_c^\top \in \mathbb{R}^{p \times c}$ with $\bar{\phi}(\mathbf{x}) = \sum_{i=1}^n \phi(\mathbf{x}_i)/n$ is the global centre of mass of the data embedded in $\mathcal{H}$.

It turns out there is an interesting connection between the occurrence of the uniform solution and the fuzzifier $m$, which is investigated in detail in Section \ref{sec:stability}. Sufficiently large values of $m$ are mathematically expected to produce maximal fuzziness, naturally driving the membership degrees toward $1/c$ for all observations across all clusters. This is not surprising given that $m$ controls the degree of fuzziness. However, a robust fuzzy clustering algorithm must approach this uniform state smoothly as $m$ increases asymptotically. If the algorithm abruptly falls into this state at finite but sufficiently small values of $m$, we say it has structurally collapsed to the uniform solution. In Figure~\ref{fig:unif_collapse_iris}, we illustrate this phenomenon of uniform collapse using the Iris data set and KFRC with the Gaussian kernel, suitably selected bandwidth values, and the power fuzzifier function.
\begin{figure}[!ht]
    \centering
    \includegraphics[width=\linewidth]{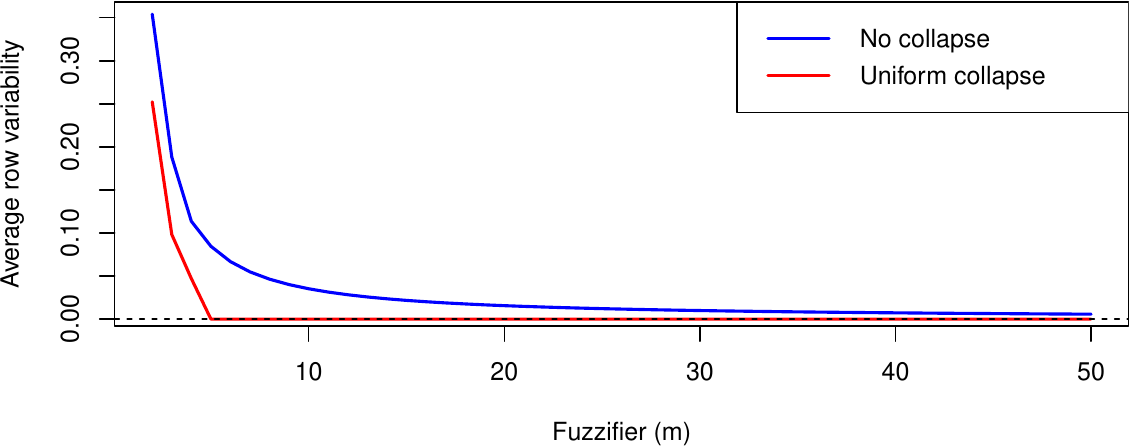}
    \caption{A comparison of a reasonable clustering and uniform collapse on the Iris data set using KFRC. Uniform collapse yields an abrupt drop of the average row variability, defined as the average standard deviation of the values across all rows of the membership matrix $\mathbf{U}$, to zero (red line), whereas no collapse enables the safe use of large values of the fuzzifier $m$ (blue line).}
    \label{fig:unif_collapse_iris}
\end{figure}

As shown in Figure~\ref{fig:unif_collapse_iris}, a poor selection of bandwidth values can lead to the uniform solution dominating for finite values of $m$, whereas the use of carefully tuned bandwidths mitigates the risk of uniform collapse, allowing the fuzziness to vary smoothly as the fuzzifier increases and removing the need for picking a value for the latter. We remark that uniform collapse was encountered for $m \geq 5$ in this example, which may be seen as an extreme choice for the value of the fuzzifier. However, this phenomenon may also occur for much smaller and commonly used fuzzifier values, in which case having a mechanism to avoid such a situation becomes crucial.

\section{Stability analysis}\label{sec:stability}

In this Section, we perform a stability analysis of the uniform solution $(\mathbf{U}^*, \mathbf{V}^*)$ defined in Section~\ref{subsec:unif_solution}. We will assume that we are performing FRC with a general fuzzifier function $t_m$ and a general dissimilarity matrix $\mathbf{R}$.

% \vspace{-0.55cm}

\subsection{The uniform solution as a fixed point}

We note that Expression~\eqref{eq:frc_general} can be equivalently written as
\begin{align*}
    J_\text{FRC}(\mathbf{U}) = \sum\limits_{i=1}^c \frac{Q_i}{2S_i},
\end{align*}
where we define $S_i \defeq \sum_{l=1}^n t_m(u_{il})$ and $Q_i \defeq \sum_{j=1}^n \sum_{k=1}^n t_m(u_{ij})t_m(u_{ik})R_{jk}$. Moreover, we set $T_{ij} \defeq \sum_{k=1}^n t_m(u_{ik})R_{jk}$. Minimisation of $J_\text{FRC}$ is equivalent to minimising the Lagrangian 
\begin{equation}\label{eq:kfrc_lagrangian}
    \mathcal{L}_\text{FRC} = \sum\limits_{i=1}^c \frac{Q_i}{2S_i} - \sum\limits_{j=1}^n \lambda_j\left( \sum\limits_{i=1}^c u_{ij} - 1\right).
\end{equation}
Differentiating Expression \eqref{eq:kfrc_lagrangian} with respect to $u_{ij}$ and setting the derivative to zero yields $u_{ij} = (t'_m)^{-1}(2\lambda_j/A_{ij})$, with $A_{ij} \defeq 2T_{ij}/S_i - Q_i/S_i^2$ and the condition $\sum_{i=1}^c u_{ij}=1$.

\begin{theorem}\label{thm:fixed_point}
    The uniform solution $(\mathbf{U}^*, \mathbf{V}^*)$ is a fixed point of the FRC algorithm.
\end{theorem}
\begin{proof}
    At the uniform solution, $u_{ij} = 1/c \ \forall i,j$, which yields
    \begin{equation*}
        Q_i = t_m^2(1/c)\sum\limits_{l=1}^n\sum\limits_{k=1}^n R_{lk}, \quad 
        S_i = nt_m(1/c), \quad
        T_{ij} = t_m(1/c) \sum\limits_{k=1}^n R_{jk}.
    \end{equation*}
    Hence, $A_{ij} = 2\sum_{k=1}^n R_{jk}/n -  \sum_{l=1}^n\sum_{k=1}^n R_{lk}/n^2$, which is independent of the cluster index $i$ and can be equivalently denoted as $A^*_j$. The first-order condition derived above gives $t'_m(1/c) = 2\lambda_j/A_j^*$ at $\mathbf{U}^*$. Therefore, setting $\lambda_j = A^*_jt'_m(1/c)/2$ proves that the stationary point of the objective function coincides with the update of FRC and thus, the uniform solution is a fixed point of the algorithm.
\end{proof}

\subsection{Stability conditions for the uniform solution}

Theorem~\ref{thm:fixed_point} proves that the uniform solution $(\mathbf{U}^*, \mathbf{V}^*)$ is always a fixed point of the FRC algorithm. Consequently, uniform initialisation will leave the algorithm trapped in this state, providing no informative insight into the underlying cluster structure. To investigate the conditions under which the algorithm successfully avoids this situation, we perform a local stability analysis. Our objective is to derive the threshold at which the uniform solution loses stability (i.e., a bifurcation point), which we formalise in the following theorem.
\begin{theorem}\label{thm:kfrc_stability_cond}
    The uniform solution $(\mathbf{U}^*, \mathbf{V}^*)$ loses stability and becomes a saddle point of the FRC algorithm when \textcolor{black}{$\lambda_\text{min}(\mathbf{Q}) <0$, where $\lambda_\text{min}(\mathbf{Q})$ is the smallest eigenvalue of the matrix $\mathbf{Q}$, defined as}
    \begin{align*}
        \mathbf{Q} & =\frac{[\chi(c)]^2}{2n\zeta(c)}\textcolor{black}{\mathbf{HRH}} + \frac{\psi(c)}{2n}\mathrm{diag}(\mathbf{r}) - \frac{\psi(c) R_{\bullet \bullet}}{4n^2}\mathbf{I}_n,\\
        \mathbf{H} & = \mathbf{I}_n - \frac{1}{n}\mathbf{1}_n\mathbf{1}_n^\top, \ \mathbf{r} = (R_{1\bullet}, \ldots, R_{n\bullet})^\top, \
        R_{j\bullet} = \sum\limits_{k=1}^n R_{jk}, \ R_{\bullet \bullet} = \sum\limits_{j=1}^n\sum\limits_{k=1}^n R_{jk}, \\
        \zeta(c) & = t_m\left(\frac{1}{c}\right), \ \chi(c) = t_m'\left( \frac{1}{c} \right), \ \psi(c) = t''_m\left(\frac{1}{c} \right).
    \end{align*}
\end{theorem}
\begin{proof}
    Suppose that we initialise FRC from a perturbed uniform solution such that $u^{(0)}_{ij} = 1/c + \epsilon\delta_{ij}$, where $\epsilon > 0$ is small, $\sum_{i=1}^c\delta_{ij} = 0$ to ensure $\sum_{i=1}^c u^{(0)}_{ij} = 1$ for all $j$. A Taylor expansion for $t_m$ around $1/c$ gives
    \begin{equation*}
        t_m\left(u^{(0)}_{ij}\right) = \zeta(c) + \chi(c) \epsilon \delta_{ij} + \frac{\epsilon^2\delta^2_{ij}}{2}\psi(c) + \mathcal{O}(\epsilon^3).
    \end{equation*}
    Therefore, we get the following expressions 
    \begin{align*}
        S_i & = n\zeta(c) + \epsilon \chi(c) \sum\limits_{j=1}^n \delta_{ij} +  \frac{\epsilon^2}{2}\psi(c) \sum\limits_{j=1}^n \delta^2_{ij} + \mathcal{O}(\epsilon^3),\\
        Q_i & = [\zeta(c)]^2R_{\bullet \bullet} + 2\epsilon\chi(c)\zeta(c) \sum\limits_{j=1}^n \delta_{ij}R_{j\bullet} + \\
    & + \epsilon^2\left\{[\chi(c)]^2 \sum\limits_{j=1}^n\sum\limits_{k=1}^n \delta_{ij}\delta_{ik}R_{jk} + \zeta(c) \psi(c) \sum\limits_{j=1}^n \delta_{ij}^2R_{j \bullet} \right\} + \mathcal{O}(\epsilon^3).
    \end{align*}
    Since $\epsilon > 0$ is small, we can expand $1/S_i$ as
    \begin{equation*}
        \frac{1}{S_i} = \frac{1}{n\zeta(c)}\left\{1 - \frac{\epsilon \chi(c)\sum\limits_{j=1}^n \delta_{ij}}{n\zeta(c)} - \frac{\epsilon^2\psi(c) \sum\limits_{j=1}^n \delta^2_{ij}}{2n\zeta(c)} + \frac{\epsilon^2[\chi (c)]^2 \left( \sum\limits_{j=1}^n \delta_{ij} \right)^2}{n^2[\zeta(c)]^2} + \mathcal{O}(\epsilon^3) \right\}.
    \end{equation*}
    Substituting the above into $J_\text{FRC}(\mathbf{U}) = \sum_{i=1}^c Q_i/(2S_i)$, the first-order term in $\epsilon$ vanishes since $\sum_{i=1}^c\delta_{ij} = 0$ for all $j$, consistent with the uniform solution being a fixed point (Theorem~\ref{thm:fixed_point}). Collecting all second-order terms in $\epsilon$ yields the following term for the Hessian
    \begin{align}\label{eq:JFRC}
        J_\text{FRC}^{(2)} = \sum_{i=1}^c\Bigg\{ & \frac{[\chi(c)]^2}{2n\zeta(c)}\sum_{j=1}^n\sum_{k=1}^n\delta_{ij}\delta_{ik}R_{jk} - \frac{[\chi(c)]^2}{n^2\zeta(c)}\Big(\sum_{j=1}^n\delta_{ij}\Big)\Big(\sum_{k=1}^n\delta_{ik}R_{k\bullet}\Big) \nonumber\\
        & + \frac{[\chi(c)]^2R_{\bullet\bullet}}{2n^3\zeta(c)}\Big(\sum_{j=1}^n\delta_{ij}\Big)^2 + \frac{\psi(c)}{2n}\sum_{j=1}^n\delta_{ij}^2R_{j\bullet} - \frac{\psi(c) R_{\bullet\bullet}}{4n^2}\sum_{j=1}^n\delta_{ij}^2\Bigg\}.
    \end{align}
    \textcolor{black}{The first three terms of Expression \eqref{eq:JFRC} are equivalent to $[\chi(c)]^2(\mathbf{H}\boldsymbol{\delta}_i)^\top \mathbf{R} (\mathbf{H}\boldsymbol{\delta}_i)/[2n\zeta(c)]$,} where $\boldsymbol{\delta}_i = (\delta_{i1}, \ldots, \delta_{in})^\top$, whereas the inner summand in the fourth term can be expressed as $\boldsymbol{\delta}_i^\top \mathrm{diag}(\mathbf{r})\boldsymbol{\delta}_i$. Hence, we can express $J_\text{FRC}^{(2)}$ as $\sum_{i=1}^c \boldsymbol{\delta}_i^\top \mathbf{Q}\boldsymbol{\delta}_i$, where
    \begin{align*}
        \mathbf{Q} = \frac{[\chi(c)]^2}{2n\zeta(c)}\mathbf{HRH} + \frac{\psi(c)}{2n}\mathrm{diag}(\mathbf{r}) - \frac{\psi(c) R_{\bullet \bullet}}{4n^2}\mathbf{I}_n.
    \end{align*}
    For the uniform solution to be an unstable fixed point (i.e. a saddle point) we require that $J_\text{FRC}^{(2)} < 0$ for some non-trivial $\boldsymbol{\delta}_1, \ldots, \boldsymbol{\delta}_c$ summing to $\mathbf{0}$. Since $J_\text{FRC}^{(2)} = \sum_i\boldsymbol{\delta}_i^\top \mathbf{Q}\boldsymbol{\delta}_i$, this holds if and only if $\boldsymbol{\gamma}^\top \mathbf{Q} \boldsymbol{\gamma} < 0$ for some non-trivial $\boldsymbol{\gamma} \in \mathbb{R}^n$. This is easily proven, since given such a vector $\boldsymbol{\gamma}$ and setting $\boldsymbol{\delta}_1 = \boldsymbol{\gamma}$, $\boldsymbol{\delta}_2 = -\boldsymbol{\gamma}$, and $\boldsymbol{\delta}_i = 0$ for $i \ge 3$ gives $J_\text{FRC}^{(2)} = 2\boldsymbol{\gamma}^\top \mathbf{Q} \boldsymbol{\gamma} < 0$, whereas if $J_\text{FRC}^{(2)} < 0$ for some admissible perturbation, then at least one summand $\boldsymbol{\delta}_i^\top \mathbf{Q} \boldsymbol{\delta}_i < 0$ and $\boldsymbol{\gamma} = \boldsymbol{\delta}_i$ is the required vector. Therefore, by a direct application of the variational characterisation of eigenvalues and since $\mathbf{Q}$ is a symmetric matrix, the existence of an admissible perturbation satisfying $J_\text{FRC}^{(2)} < 0$ is equivalent to the minimum eigenvalue of $\mathbf{Q}$ being negative, that is, $\lambda_\text{min}(\mathbf{Q}) < 0$. \textcolor{black}{Finally, since $\psi(c) > 0$ by convexity of $t_m$ and $R_{\bullet\bullet} > 0$ (unless all $n$ objects are identical), we have $\mathbf{1}_n^\top\mathbf{Q}\mathbf{1}_n = \psi(c)R_{\bullet\bullet}/(4n) > 0$. Hence the same construction with $\boldsymbol{\gamma} = \mathbf{1}_n$ gives an admissible perturbation with $J^{(2)}_{\text{FRC}} > 0$, so the uniform solution is not a local maximum and is therefore a saddle point.}
\end{proof}
\begin{corollary}\label{cor:euclid_stability_cond}
    Let $\mathbf{K} \in \mathbb{R}^{n \times n}$ be a Gram matrix induced by a symmetric PSD kernel function $k: \mathcal{X} \times \mathcal{X} \rightarrow \mathbb{R}$, and let $\mathbf{R} \in \mathbb{R}^{n \times n}$ be the corresponding dissimilarity matrix with entries $R_{jk} = K_{jj} + K_{kk} - 2K_{jk}$. Then, the uniform solution $(\mathbf{U}^*, \mathbf{V}^*)$ loses stability when $\kappa(t_m, 1/c)\lambda_\text{max}(\mathbf{F}) > 1,$ where $\kappa(t_m, 1/c) = 2[\chi(c)]^2/[\zeta(c)\psi(c)]$ and
    \begin{equation*}
        \mathbf{F} = \frac{1}{n}\mathbf{D}^{-1/2}\bar{\mathbf{K}}\mathbf{D}^{-1/2}, \quad 
        \bar{\mathbf{K}} = \mathbf{H}\mathbf{K}\mathbf{H}, \quad \mathbf{D} = \mathrm{diag}(\bar{\mathbf{K}}),
    \end{equation*}
    where we assume $\mathbf{D} \succ 0$, i.e. $\mathbf{D}$ is positive definite.
\end{corollary}
\begin{proof}
    This result follows directly from Theorem~\ref{thm:kfrc_stability_cond} and the fact that the dissimilarity matrix entries can be written as $R_{jk} = K_{jj} + K_{kk} - 2K_{jk}$. \textcolor{black}{Since $R_{jk} = \|\phi(\mathbf{x}_j) - \phi(\mathbf{x}_k)\|^2_{\mathcal{H}}$ is unchanged when each feature vector is centred, i.e. when $\phi(\mathbf{x}_j)$ is replaced by $\tilde\phi(\mathbf{x}_j) = \phi(\mathbf x_j) - \bar\phi(\mathbf{x})$, we have $R_{jk} = \bar K_{jj} + \bar K_{kk} - 2\bar K_{jk}$, where $\phi$ is the feature map associated with the kernel function $k$. In matrix notation, this is $\mathbf{R} = \mathbf{d}\mathbf{1}_n^\top + \mathbf{1}_n\mathbf{d}^\top - 2\bar{\mathbf{K}}$, with $\mathbf{d} = \mathrm{diag}(\bar{\mathbf{K}})$. Consequently, using $\mathbf{H}\mathbf{1}_n = \mathbf{0}$, $\bar{\mathbf{K}}\mathbf{1}_n = \mathbf{0}$ and $\mathbf{H}\bar{\mathbf{K}}\mathbf{H} = \bar{\mathbf{K}}$, the sum terms in Theorem~\ref{thm:kfrc_stability_cond} simplify to
    \begin{align*}
        \mathbf{HRH} &= -2\bar{\mathbf{K}},\\
        r_j & = \sum\limits_{k=1}^n \left(\bar{K}_{jj} + \bar{K}_{kk} - 2\bar{K}_{jk}\right) = n\bar{K}_{jj} + \text{Tr}(\bar{\mathbf{K}}) \implies \mathrm{diag}(\mathbf{r}) = n\mathbf{D} + \text{Tr}(\bar{\mathbf{K}})\mathbf{I}_n, \\
        R_{\bullet \bullet} & = \sum\limits_{j=1}^n r_j = 2n\text{Tr}(\bar{\mathbf{K}}).
    \end{align*}
    Substituting these into the definition of the matrix $\mathbf{Q}$ yields
    \begin{align*}
        \mathbf{Q} & = \frac{[\chi(c)]^2}{2n\zeta(c)}(- 2\bar{\mathbf{K}}) + \frac{\psi(c)}{2n}[n\mathbf{D} + \text{Tr}(\bar{\mathbf{K}})\mathbf{I}_n] - \frac{2n\psi(c)\text{Tr}(\bar{\mathbf{K}})}{4n^2}\mathbf{I}_n \\
        & = -\frac{[\chi(c)]^2}{n\zeta(c)}\bar{\mathbf{K}} + \frac{\psi(c)}{2}\mathbf{D}.
    \end{align*}
    To evaluate stability, take any non-trivial vector $\boldsymbol{\gamma} \in \mathbb{R}^n$ so that
    \begin{equation*}
        \boldsymbol{\gamma}^\top \mathbf{Q} \boldsymbol{\gamma} = \boldsymbol{\gamma}^\top \left\{ - \frac{[\chi(c)]^2}{n\zeta(c)}\bar{\mathbf{K}} + \frac{\psi(c)}{2}\mathbf{D} \right\} \boldsymbol{\gamma}
    \end{equation*}
    and set $\mathbf{y} = \mathbf{D}^{1/2}\boldsymbol{\gamma}$ to obtain
    \begin{align*}
        \boldsymbol{\gamma}^\top \mathbf{Q} \boldsymbol{\gamma} & = \mathbf{y}^\top \left\{ - \frac{[\chi(c)]^2}{\zeta(c)}\mathbf{F} + \frac{\psi(c)}{2} \mathbf{I}_n \right\} \mathbf{y}.
    \end{align*}
    Since $\mathbf{D} \succ 0$, $\mathbf{D}^{1/2}$ is invertible and therefore $\mathbf{y}$ ranges over all non-trivial vectors in $\mathbb{R}^n$ as $\boldsymbol{\gamma}$ does.} Hence, by Theorem~\ref{thm:kfrc_stability_cond}, the uniform solution loses stability if and only if the minimum eigenvalue of $-[\chi(c)]^2\mathbf{F}/\zeta(c) + \psi(c)\mathbf{I}_n/2$ is strictly negative. Since $\zeta(c) = t_m(1/c) > 0$ (from $t_m(0) = 0$ and $t_m' > 0$), we have $[\chi(c)]^2/\zeta(c) > 0$, so this minimum eigenvalue is attained at the eigenvector of $\lambda_{\max}(\mathbf{F})$, and the condition becomes
    \begin{equation*}
        \frac{\psi(c)}{2} - \frac{[\chi(c)]^2}{\zeta(c)}\lambda_{\max}(\mathbf F) < 0.
    \end{equation*}
    Since $\psi(c) > 0$, rearranging gives
    \begin{equation*}
        \frac{2[\chi(c)]^2}{\zeta(c)\psi(c)}\lambda_{\max}(\mathbf{F}) > 1 \iff \kappa(t_m, 1/c)\,\lambda_{\max}(\mathbf{F}) > 1,
    \end{equation*}
    where $\kappa(t_m, u) = 2[t_m'(u)]^2/[t_m(u)t_m''(u)]$, as stated.
\end{proof}
Corollary~\ref{cor:euclid_stability_cond} is a generalisation of the rule derived in \cite{yu2004analysis}, where it was shown that the FCM algorithm does not collapse to the uniform solution as long as $\lambda_\text{max}(\mathbf{F}) \geq 0.5$ or when the fuzzifier is chosen so that $m < 1/(1-2\lambda_\text{max}(\mathbf{F}))$ if $\lambda_\text{max}(\mathbf{F}) < 0.5$. In order to see this, one simply needs to consider the linear kernel $k(\mathbf{x}, \mathbf{x}) = \mathbf{x}^\top \mathbf{x}$ to recover squared Euclidean distances and use the power fuzzifier $t_m(u) = u^m$.

\textcolor{black}{Theorems~\ref{thm:fixed_point} and \ref{thm:kfrc_stability_cond} hold for general fuzzy relational clustering. Theorem~\ref{thm:fixed_point} requires only a valid fuzzifier function, and Theorem~\ref{thm:kfrc_stability_cond} additionally requires $\mathbf{R}$ to be symmetric and $t_m$ to be twice continuously differentiable around $1/c$. No Euclidean or Hilbert-space structure is assumed in either. Corollary~\ref{cor:euclid_stability_cond} is the special case of dissimilarities induced by a symmetric PSD kernel, which is the setting of KFRC. In this case, the stability condition of Theorem~\ref{thm:kfrc_stability_cond} reduces to a condition on the leading eigenvalue of $\mathbf{F}$, which can be targeted directly through the kernel bandwidths. All results are local, in the sense that they concern the behaviour of the objective in a neighbourhood of the uniform solution.}

This stability framework can be naturally extended to predict the fuzziness threshold at which distinct clusters begin to merge. Conceptually, a cluster merge can be viewed as a phase transition, mirroring the thermodynamic bifurcations examined by Rose \cite{rose2002deterministic} in the context of Deterministic Annealing. To approximate this critical value, we first execute FRC with a sufficiently low fuzzifier $m$ to ensure that all $c$ clusters are represented and harden the resulting membership probabilities. We then construct data sub-matrices for every pair of clusters and apply the stability condition from Theorem~\ref{thm:kfrc_stability_cond} locally, assuming $c = 2$. The smallest value of $m$ across all pairs that renders this local two-cluster uniform solution a stable fixed point is an upper bound for the fuzzifier value at which the first cluster merge occurs in the global system.

\subsection{A rule for avoiding uniform collapse}

Given the stability analysis conducted, we can obtain conditions for which the KFRC algorithm does not collapse to the uniform solution even for large values of the fuzzifier $m$. We focus specifically on symmetric PSD kernels for which we have an analytic expression for the stability condition (see Corollary~\ref{cor:euclid_stability_cond}). Assuming the fuzzifier parameter $m$ takes values in $[m^\text{min}, m^\text{max}]$, the uniform solution is never a stable fixed point for any valid $m$ if $\lambda_\text{max}(\mathbf{F}) \geq \lim_{m \rightarrow m^\text{max}} 1/\kappa(t_m, 1/c)$. This establishes a lower bound for the leading eigenvalue of $\mathbf{F}$ that depends on the choice of the fuzzifier function: for the power and exponential fuzzifiers, this bound is equal to 1/2 \cite{klawonn2003alternative}, whereas for the quadratic fuzzifier it is 1/4 \cite{klawonn2003fuzzy}. Because $\mathbf{F}$ is constructed via double-centring, its rank is at most $n-1$. Furthermore, the eigenvalues of $\mathbf{F}$ sum to one, which can be seen by
\begin{equation*}
    \sum_{i=1}^n \lambda_i(\mathbf{F}) = \text{Tr}(\mathbf{F}) = \frac{1}{n}\text{Tr}(\mathbf{D}^{-1/2}\bar{\mathbf{K}}\mathbf{D}^{-1/2}) = \frac{1}{n}\text{Tr}(\bar{\mathbf{K}}\mathbf{D}^{-1}) = \frac{1}{n}\sum\limits_{i=1}^n \frac{\bar{\mathbf{K}}_{{ii}}}{\mathbf{D}_{ii}} = 1.
\end{equation*}
\textcolor{black}{Moreover, note that the matrix $n\mathbf{F}$ represents the kernel Gram matrix that has been double-centred and normalised to the unit hypersphere, therefore the entries of $\mathbf{F}$ can be written as}
\begin{equation*}
(\mathbf{F})_{i,j} = \frac{1}{n}\frac{\langle \tilde{\phi}(\mathbf{x}_i), \tilde{\phi}(\mathbf{x}_j) \rangle_{\mathcal{H}}}{\| \tilde{\phi}(\mathbf{x}_i)\|_{\mathcal{H}} \|\tilde{\phi}(\mathbf{x}_j)\|_{\mathcal{H}}} = \frac{1}{n}\cos(\omega_{ij}),
\end{equation*}
where $\tilde{\phi}(\mathbf{x}_i) = \phi(\mathbf{x}_i) - \bar{\phi}(\mathbf{x})$ is the centred projection of $\mathbf{x}_i$ in the RKHS, and $\omega_{ij}$ is the angle between $\tilde{\phi}(\mathbf{x}_i)$ and $\tilde{\phi}(\mathbf{x}_j)$. By direct analogy to Kernel PCA, the leading eigenvalue of $\mathbf{F}$ is interpreted as the proportion of variance explained by the first principal component on the unit hypersphere defined in the kernel-induced feature space. The condition derived in Corollary~\ref{cor:euclid_stability_cond} is thus equivalent to requiring that the angular concentration along a single principal axis be large enough to ensure that the algorithm never converges to the uniform solution, regardless of the fuzzifier $m$. However, this geometric requirement exposes a critical limitation of existing fuzzifiers. If a dataset contains $c$ well-separated clusters, the centroids naturally span a $(c-1)$-dimensional subspace. Using the power or exponential fuzzifier necessitates that half of the geometric variability be explained by a single dimension. This restricts the remaining $c-2$ dimensions to share the remaining variance, effectively distorting the cluster allocation process for $c > 2$. To rectify this, we argue that an additional desirable property for a fuzzifier function, besides origin sparsity, is an asymptotic scaling of the stability threshold of
\begin{equation*}
        \lim\limits_{c \rightarrow \infty}\left[\lim\limits_{m \rightarrow m^\text{max}} \frac{1}{\kappa(t_m, 1/c)}\right] = 0.
\end{equation*}
This ensures that as the number of clusters increases, the required variance bound on the leading eigenvalue scales down appropriately, preventing artificial dimensional dominance. To meet this requirement, we introduce the complementary root fuzzifier function, $t_{\text{cr}}$, which satisfies all standard properties of fuzzifier functions alongside the desirable properties of origin sparsity and of asymptotic geometric stability
\begin{equation*}
    t_{\text{cr}}: [0,1] \rightarrow [0,1], \quad t_\text{cr}(u) = 1 - (1-u)^{1/m}, \quad m > 1.
\end{equation*}
A formal proof verifying that the complementary root fuzzifier satisfies all necessary and desirable conditions is provided in \ref{appen4}.

\section{Proposed bandwidth selection algorithm}\label{sec:bw_select_alg}

Building upon the theoretical guarantees established in Section~\ref{sec:stability}, we propose a two-stage bandwidth selection strategy. The primary objective of this approach is twofold. First, it ensures the selected bandwidths decouple the fuzzifier parameter from the risk of uniform collapse, thereby guaranteeing the uniform solution remains unstable. Second, it optimises a multi-cluster criterion designed to isolate relevant cluster structures by assigning larger bandwidth values to uninformative features, effectively smoothing them out. Throughout this Section, we assume the use of a symmetric PSD kernel as this allows us to directly leverage the analytic bound for the leading eigenvalue of matrix $\mathbf{F}$ derived in Corollary~\ref{cor:euclid_stability_cond} to mathematically prevent uniform collapse.

In order to implement both stages of this bandwidth selection strategy, we must navigate the continuous parameter space of the kernel bandwidths. Given the formulation of our objective criteria, we approach this as a continuous optimisation problem and employ multi-start gradient ascent. The viability of any gradient-based approach in this context hinges on the differentiability of the eigenvalue map with respect to the kernel parameters, a property we formalise below.

\begin{proposition}\label{prop:diff_ftheta}
    Let $\mathbf{F} \equiv \mathbf{F}(\boldsymbol{\theta}) = \mathbf{D}^{-1/2}\bar{\mathbf{K}}\mathbf{D}^{-1/2}/n$ with $\mathbf{D}, \bar{\mathbf{K}}$ defined as in Corollary~\ref{cor:euclid_stability_cond} and assume a symmetric PSD kernel is used with bandwidths $\boldsymbol{\theta} = (\theta_1, \ldots, \theta_p)^\intercal$. Then, if the leading eigenvalue of $\mathbf{F}$ has multiplicity one, it is differentiable with respect to $\boldsymbol{\theta}$.
\end{proposition}

\begin{proof}
    Since we are using a PSD kernel that is symmetric, we know that the kernel Gram matrix $\mathbf{K}$ is symmetric and has real values, thus it is Hermitian. Symmetry is preserved for $\mathbf{F}$, making it Hermitian as well. Based on Weyl's inequality, eigenvalue maps are Lipschitz continuous on Hermitian matrices. Under the additional assumption of unit multiplicity for the leading eigenvalue of $\mathbf{F}$, differentiability of $\lambda_1(\mathbf{F}(\boldsymbol{\theta}))$ is guaranteed.
\end{proof}

\textcolor{black}{When the leading eigenvalue of $\mathbf{F}$ has multiplicity greater than one, $\lambda_1(\mathbf{F}(\boldsymbol{\theta}))$ is not differentiable. It is, however, Lipschitz and regular, and for any unit eigenvector $\mathbf{v}$ associated with $\lambda_1(\mathbf{F})$, the vector $\big[\mathbf{v}^\top (\partial\mathbf{F}/ \partial\theta_j)\mathbf{v}\big]_{j=1}^p$ belongs to its generalised gradient (see Example 2.8.7 and Proposition 2.8.8 in \cite{clarke1990}). Hence, the updates in Stage 1 remain well-defined. Moreover, the feasibility condition $\lambda_1(\mathbf{F}) \geq \tau$ is verified directly in \textproc{ProjFeas} rather than through its gradient, so feasibility is unaffected. Note that the multi-cluster criteria in Stage 2 may also fail to be differentiable where eigenvalues cross. Such points are rarely hit exactly in practice, and since Stage 2 is run from multiple random starts and the best solution is retained, their effect on the final bandwidths is limited.}

We provide the relevant gradient expressions below, noting that $\lambda_1(\mathbf{F}) \equiv \lambda_1(\mathbf{F}(\boldsymbol{\theta}))$ refers to the leading eigenvalue of the matrix $\mathbf{F}$ when the vector of bandwidths is $\boldsymbol{\theta}$. Similarly, the notation $\mathbf{A}(\boldsymbol{\theta})$ refers to a matrix $\mathbf{A}$ when the vector of bandwidths is $\boldsymbol{\theta}$, with gradient
\begin{align*}
    \frac{\mathrm{d} \lambda_1(\mathbf{F})}{\mathrm{d} \boldsymbol{\theta}} & = \mathbf{v}^\intercal \frac{\partial\mathbf{F}(\boldsymbol{\theta})}{\partial\boldsymbol{\theta}} \mathbf{v},
\end{align*}
where $\mathbf{v}$ is the eigenvector that corresponds to $\lambda_1(\mathbf{F})$. By the chain rule, the derivative of $\mathbf{F}(\boldsymbol{\theta}) \equiv \mathbf{F}$ with respect to $\boldsymbol{\theta}$ is
\begin{equation*}
    \frac{\partial \mathbf{F}}{\partial\boldsymbol{\theta}} = \frac{\partial \mathbf{F}}{\partial \bar{\mathbf{K}}(\boldsymbol{\theta})}\frac{\partial\bar{\mathbf{K}}(\boldsymbol{\theta})}{\partial \boldsymbol{\theta}} + \frac{\partial \mathbf{F}}{\partial \mathbf{D}(\boldsymbol{\theta})}\frac{\partial \mathbf{D}(\boldsymbol{\theta})}{\partial \boldsymbol{\theta}}.
\end{equation*}
From this point onwards we concentrate on derivatives with respect to the $j$th bandwidth parameter $\theta_j$. The Expressions for the individual terms involved in the chain rule are included in \ref{appen1}.

The bandwidth selection procedure consists of two steps, which are presented in Algorithms \ref{alg:stage1} \& \ref{alg:stage2} in \ref{appen2}. The first stage of the algorithm returns a vector of feasible bandwidths, with the feasibility criterion requiring the leading eigenvalue of $\mathbf{F}$ to exceed a stability threshold value $\tau = \lim_{m \rightarrow m^\text{max}} 1/\kappa(t_m, 1/c)$. More precisely, Stage 1 uses gradient ascent to find a vector of bandwidths that maximise the leading eigenvalue of $\mathbf{F}$, starting from an initial estimate of the variability of each feature plus some noise. This is repeated $S$ times, with different starting points returning at most $S$ candidates. Convergence is assessed by whether the Chebyshev distance between subsequent iterations exceeds a small tolerance $\varepsilon$ and terminates after a user-specified maximum number of iterations $T$. Stage 1 returns at most $S$ candidate bandwidth vectors, discarding any solutions for which the leading eigenvalue of $\mathbf{F}$ is below the stability threshold $\tau$.

Stage 2 is a refinement step and aims to `correct' the feasible bandwidth vectors obtained from Stage 1, so that a multi-cluster criterion $\mathcal{Q}$ is maximised. Examples of such criteria are the sum or the sum of the logs of the first $c-1$ eigenvalues of $\mathbf{F}$ (the second is referred to as the log-sum criterion), or the eigengap (that is, the difference between the $(c-1)$st and the $c$th eigenvalues of $\mathbf{F}$). Starting from each feasible candidate vector of bandwidths and the associated $\mathbf{F}$ matrix, we use gradient ascent to maximise $\mathcal{Q}$. The final vector of bandwidths returned is the one for which $\mathcal{Q}$ is maximised.

The justification behind implementing Step 2 is that Step 1 is blindly searching for a vector of bandwidths for which $\lambda_1(\mathbf{F})$ is maximised. However, this may be done at the cost of setting all but one of the bandwidths extremely large, thus oversmoothing all but one features. This will eventually lead to $\lambda_1(\mathbf{F}) > \tau$, so that all the variance in the RKHS is concentrated along a single direction, effectively collapsing all the data onto a line. Consequently, the centroids would be forced into a nearly collinear configuration, which cannot faithfully represent $c > 2$ distinct groups. Maximising a multi-cluster objective alleviates this issue as it directly promotes a configuration where the data are strongly aligned along the directions that separate the clusters, without artificially inflating the variance in irrelevant dimensions.

Notice that Stage 1 and Stage 2 involve two additional functions, namely \textproc{ProjAniso} and \textproc{ProjFeas} (Algorithms \ref{alg:proj_aniso} \& \ref{alg:proj_feas} in \ref{appen2}, respectively). The former imposes an anisotropy constraint and ensures that the proposed bandwidths at each step of the optimisation procedure satisfy $\max_j \theta_j / \min_j \theta_j \le \alpha$, where $\alpha \geq 1$. Setting $\alpha = 1$ forces the use of an isotropic kernel, whereas $\alpha \rightarrow \infty$ allows for any vector of bandwidths to be proposed. We recommend setting $1 < \alpha <\infty$, just so that the maximisation of $\lambda_1$ no longer risks collapsing the data onto a single dimension and allows for uninformative features to be smoothed out. The latter algorithm guarantees that the refined bandwidths at Stage 2 remain within the feasible region defined by $\lambda_1(\mathbf{F}) \ge \tau - \delta$, for a small $\delta$ value of e.g. $10^{-6}$. This is achieved by a backtracking line search with an Armijo-like condition \cite{armijo1966minimization} that enforces feasibility. The full two-stage bandwidth selection procedure is summarised in Algorithm \ref{alg:main}.

Finally, we briefly address the computational complexity of the proposed bandwidth selection algorithm. The primary computational bottleneck in both stages is the evaluation of the spectral properties of the $n \times n$ normalised kernel matrix $\mathbf{F}$ at each gradient ascent iteration. A naive, full eigendecomposition requires $\mathcal{O}(n^3)$ operations, which becomes prohibitive for large datasets. However, our design relaxes this requirement. Stage 1 only requires the computation of the leading eigenpair $(\lambda_1, \mathbf{v}_1)$, while Stage 2 requires at most the top $c$ extremal eigenpairs. Because the number of clusters is typically much smaller than the number of observations ($c \ll n$), we can bypass a full decomposition. Instead, we employ efficient iterative approximation techniques, such as the Lanczos algorithm \cite{lanczos1950iteration} that computes all necessary eigenpairs in $\mathcal{O}(cn^2)$ operations, reducing the computational overhead and ensuring the bandwidth search remains scalable in practice.

\begin{algorithm}[t]
\caption{Two-Stage Bandwidth Selection Procedure}
\label{alg:main}
\begin{algorithmic}[1]
\Require Data $\mathbf{X} \in \mathbb{R}^{n \times p}$, number of clusters $c \ge 2$, anisotropy bound $\alpha \ge 1$, number of random starts $S$, learning rate $\eta$, momentum $\gamma \in [0, 1)$, tolerance $\varepsilon$, max iterations $T$, stability threshold $\tau > 0$, symmetric PSD kernel function $k$, multi-cluster objective $\mathcal{Q}$.
\Ensure Optimal bandwidth vector $\boldsymbol{\theta}^* \in \mathbb{R}_{>0}^p$, eigendecomposition of $\mathbf{F}(\boldsymbol{\theta}^*)$.
\Statex $\triangleright$ \textbf{Stage 1:} Generate a set of feasible candidates
\State $\mathcal{F} \gets \Call{Stage1}{\mathbf{X}, \alpha, S, \eta, \gamma, \varepsilon, T, \tau, k}$
\If{$\mathcal{F} = \varnothing$}
    \State \Return \texttt{infeasible} (no bandwidth satisfies $\lambda_1 \ge \tau$)
\EndIf
\Statex $\triangleright$ \textbf{Stage 2:} Refine each candidate and select the best
\State $\mathcal{R} \gets \varnothing$
\ForEach{$(\ell(\boldsymbol{\theta})_0, \mathbf{F}_0) \in \mathcal{F}$}
    \State $(\boldsymbol{\theta}, \mathcal{Q}) \gets \Call{Stage2}{\mathbf{X}, c, \alpha, \ell(\boldsymbol{\theta})_0, \mathbf{F}_0, \eta, \gamma, \varepsilon, T, \tau, k, \mathcal{Q}}$
    \State $\mathcal{R} \gets \mathcal{R} \cup \{(\boldsymbol{\theta}, \mathcal{Q})\}$
\EndForEach
\State $(\boldsymbol{\theta}^*, \mathcal{Q}^*) \gets \arg\max_{(\boldsymbol{\theta}, \mathcal{Q}) \in \mathcal{R}} \mathcal{Q}$
\State Compute eigendecomposition of $\mathbf{F}(\boldsymbol{\theta}^*)$.
\State \Return $\boldsymbol{\theta}^*$, $\mathcal{Q}^*$, eigendecomposition of $\mathbf{F}(\boldsymbol{\theta}^*)$
\end{algorithmic}
\end{algorithm}

\section{Simulations}\label{sec:simulations}

In this Section we describe and analyse the results of simulations on synthetic and publicly available data sets. Our proposed method (KFRC) is benchmarked against Fuzzy Relational Clustering (FRC), FANNY \cite{kaufman2009finding}, FCM with Multiple Kernels (FCM-MK) \cite{baili2011fuzzy}, Membership Scaling FCM (MSFCM) \cite{zhou2020new}, and FCM by varying the fuzziness parameter (vFCM) \cite{chen2022improved}. While there is a plethora of fuzzy clustering methods proposed, we select those due to their popularity and the similarities they share with our approach (e.g. the use of kernels, or the attempt of tuning the fuzzifier parameter). The code to reproduce the simulations is available in \url{https://github.com/EfthymiosCosta/KFRC_Sims}.

\subsection{Synthetic data}\label{subsec:syntheticdata}

Synthetic datasets are generated as mixtures of homogeneous spherical or elliptical Gaussian components with varying degrees of average pairwise overlap between clusters, defined as misclassification probabilities \cite{maitra2010simulating}. Each dataset consists of $n = 200$ observations in $p = 10$ dimensions. The first $p_s$ dimensions contain the cluster structure with the remaining $p_n = 10 - p_s$ dimensions being independent standard Gaussian noise variables appended to the signal block. The signal-to-noise ratio is reported as $\text{SNR} = p_s/p_n \in \{1/4, 1, 4, \infty \}$, corresponding to $p_s = \{2, 5, 8, 10\}$. The data sets contain $c = 2$, $4$, or $6$ clusters, with average pairwise overlap values of $0.001, 0.01, 0.05$, and $0.10$, corresponding to very low, low, moderate, and high overlap. We also assess the cases of balanced clusters and imbalanced clusters with one group being substantially smaller than the rest, containing just $5\%$ of the observations. This yields a total of 192 data sets, with 25 seeds being used to generate 4800 synthetic data sets.

Bandwidth selection for KFRC was conducted under several configurations: both sums and products of Gaussian kernels were used, with the anisotropy bound being set to $\alpha \in \{1, 5, 10, 25, \infty \}$. Three multi-cluster criteria were used: the sum of the leading $c-1$ eigenvalues, the sum of their logarithms (log-sum), and the eigengap heuristic. We used $S = 5$ random starts and allowed for at most up to $T=500$ iterations, with learning rate $\eta = 0.2$, momentum $\gamma = 0.9$, and convergence threshold $\varepsilon = 10^{-5}$. Four fuzzifier functions were used for clustering: the power, exponential, complementary root, and quadratic fuzzifiers. We set the common heuristic choice of $m = 2$ for the first three fuzzifier functions and $m = 0.5$ for the fourth. FRC was run under the exact same settings using the squared Euclidean distance, whereas FANNY was also implemented using the Manhattan distance and just the power fuzzifier. We further implemented FCM-MK, MSFCM, and vFCM. We assumed four Gaussian kernels with initially uniform weights for FCM-MK and set $m_0 = 2$, $a = 0.95$, and $b = 0.05$ for vFCM in order to update the fuzzifier as $m_{t+1} = am_t + b$ every $k = 2$ iterations, following the authors' recommendations. All clustering methods ran with a maximum of 200 iterations until convergence.

Cluster recovery is assessed via the hard Adjusted Rand Index (ARI) \cite{hubert1985comparing} and Adjusted Mutual Information (AMI) \cite{vinh_information_2010} between the hardened cluster and the true labels. Two variants of the fuzzy adjusted Rand index are reported: $\text{FARI}_\text{crisp}$ between the obtained partition matrix and the true labels \cite{campello2007fuzzy} and the Frobenius ARI ($\text{FARI}_\text{fuzzy}$) between the true component posteriors derived from the Gaussian mixture parameters and the obtained partition matrix \cite{andrews2022assessments}. Variable selection is diagnosed by the inverse Simpson index $W_{\text{eff}} = 1/\sum_j W_j^2$, where $W_j = (1/\theta_j)/\sum_{k=1}^p (1/\theta_k)$ is the kernel importance weight of dimension $j$, derived from the inverse of the fitted bandwidth. Finally, we assess convergence to the uniform solution using the partition coefficient $\mathrm{PC} = (\sum_{j=1}^n\sum_{i=1}^c u^2_{ij})/n$ \cite{bezdek1973cluster}, which equals one for perfectly crisp partitions and $1/c$ for uniformity. To make this diagnostic comparable across values of $c$, we report the uniformity index $U = (1-\text{PC})/(1-1/c)$, normalised so that zero corresponds to a fully crisp partition and one to the uniform partition.

\begin{figure}[!ht]
    \centering
    \includegraphics[width=\linewidth]{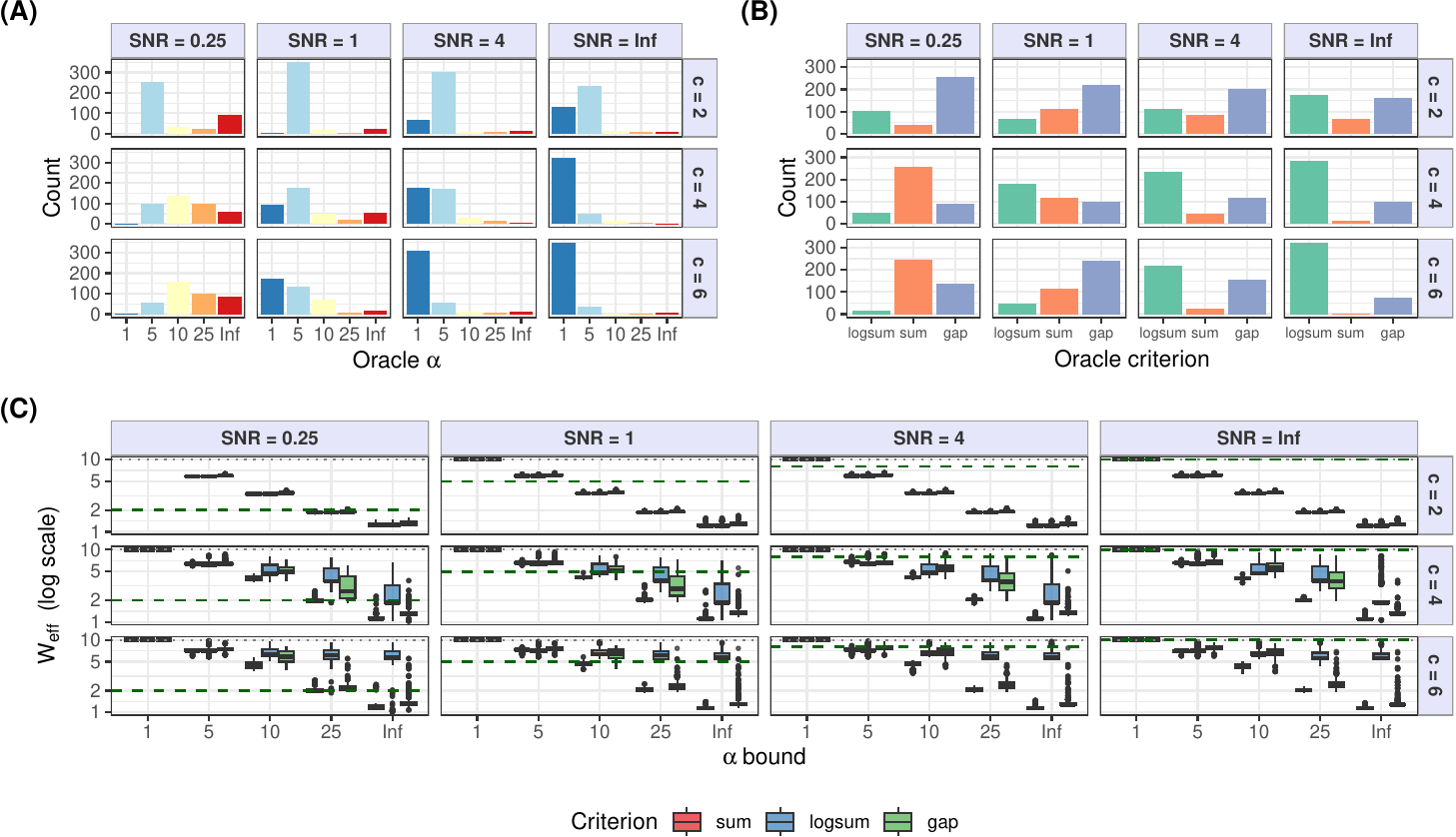}
    \caption{Top row: Bar chart of oracle anisotropy bound $\alpha$ \textbf{(A)} and oracle multi-cluster criterion \textbf{(B)} for varying $c$ and SNR. Bottom row \textbf{(C)}: Effective number of variables $W_{\text{eff}}$ (log scale) for the three criteria, at different $\alpha$ bounds per $c$ and SNR. The dotted grey and dashed green lines mark $W_{\text{eff}} = p$ and $W_{\text{eff}} = p_s$.}
    \label{fig:oracle_res_plot}
\end{figure}

The results of our simulation study are reported in Table \ref{tab:simresults} in \ref{appen3}. Overall, it is clear that KFRC with the quadratic and the complementary root fuzzifier functions, as well as vFCM, are the most successful methods in terms of cluster recovery performance. More precisely, KFRC outperforms all other methods for $\mathrm{SNR} < 1$ while maintaining solid performance for higher SNR values and for $c \geq 4$. The quadratic and complementary root fuzzifiers tend to outperform the power and exponential ones, which we attribute to their smaller stability thresholds. This also explains why the quadratic fuzzifier is more suitable for $c=2$ clusters than the complementary root. FANNY, FCM-MK and MSFCM perform suboptimally, with most of them returning uniformity scores exceeding 0.90, indicating a nearly uniform solution. This is due to the heuristic choice of $m=2$ apparently being too high for these methods to prevent uniform collapse. The closest competitor of KFRC in terms of cluster recovery is vFCM, especially when $\mathrm{SNR} \geq 1$; however, uniformity scores for vFCM are all near-zero, particularly when $c>2$. This indicates almost crisp partitions, which is unsurprising given that vFCM anneals the fuzzifier towards one. However, such an outcome is undesirable when performing fuzzy clustering and proper fuzzy solutions are expected; KFRC does not exhibit this pathology, with $U$ remaining predominantly in the intermediate range $[0.1, 0.8]$. \textcolor{black}{The uniformity index only indicates whether a partition is degenerate, that is, close to uniform or crisp. Whether the memberships are meaningful is assessed by $\mathrm{FARI}_\mathrm{fuzzy}$, which compares them with the true posterior membership probabilities.}

We also look at the optimal anisotropy bounds and criteria per SNR and $c$, and how well the bandwidth selection strategy performs at selecting the right variables. These seem to be the most influential factors in the performance of KFRC, with the choice between sums or products being negligible, whereas cluster overlap or cluster imbalance are known to deteriorate clustering quality and are therefore not of particular interest \cite{costa2023benchmarking}. Clearly, a bit of anisotropy such as $\alpha = 5$ is optimal when $c=2$, as shown in Figure~\ref{fig:oracle_res_plot}(A). As SNR increases, allowing $\alpha \rightarrow \infty$ is suboptimal and in fact, $\alpha = 1$ is preferred when all features are informative, which is an expected result. Interestingly, higher $\alpha$ values such as 10 or 25 only seem to be preferred when there is a low SNR and $c \geq 4$ clusters. In this setting, the sum criterion yielded the best $\text{FARI}_\text{fuzzy}$ results, as can be seen from Figure~\ref{fig:oracle_res_plot}(B). When the number of clusters is small ($c<4$), the eigengap criterion proves optimal regardless of the SNR. Conversely, the log-sum criterion tends to yield superior performance for $c \geq 4$ under high SNR. The efficacy of the log-sum criterion in these cases stems from its statistical interpretation as the generalised variance \cite{wilks1932certain}. Geometrically, optimising this objective is equivalent to maximising the squared volume of the parallelepiped spanned by the principal axes of the data in the kernel-induced feature space, thereby enforcing maximal cluster separation. Lastly, Figure~\ref{fig:oracle_res_plot}(C) illustrates that the choice of a suitable anisotropy bound $\alpha$ is more important overall than the multi-cluster criterion. However, the latter is critical when $\text{SNR} < 1$ and $c \geq 4$. In such cases, the log-sum criterion overestimates the number of informative features, even for $\alpha \rightarrow \infty$.

\begin{figure}[!ht]
    \centering
    \includegraphics[width=\linewidth]{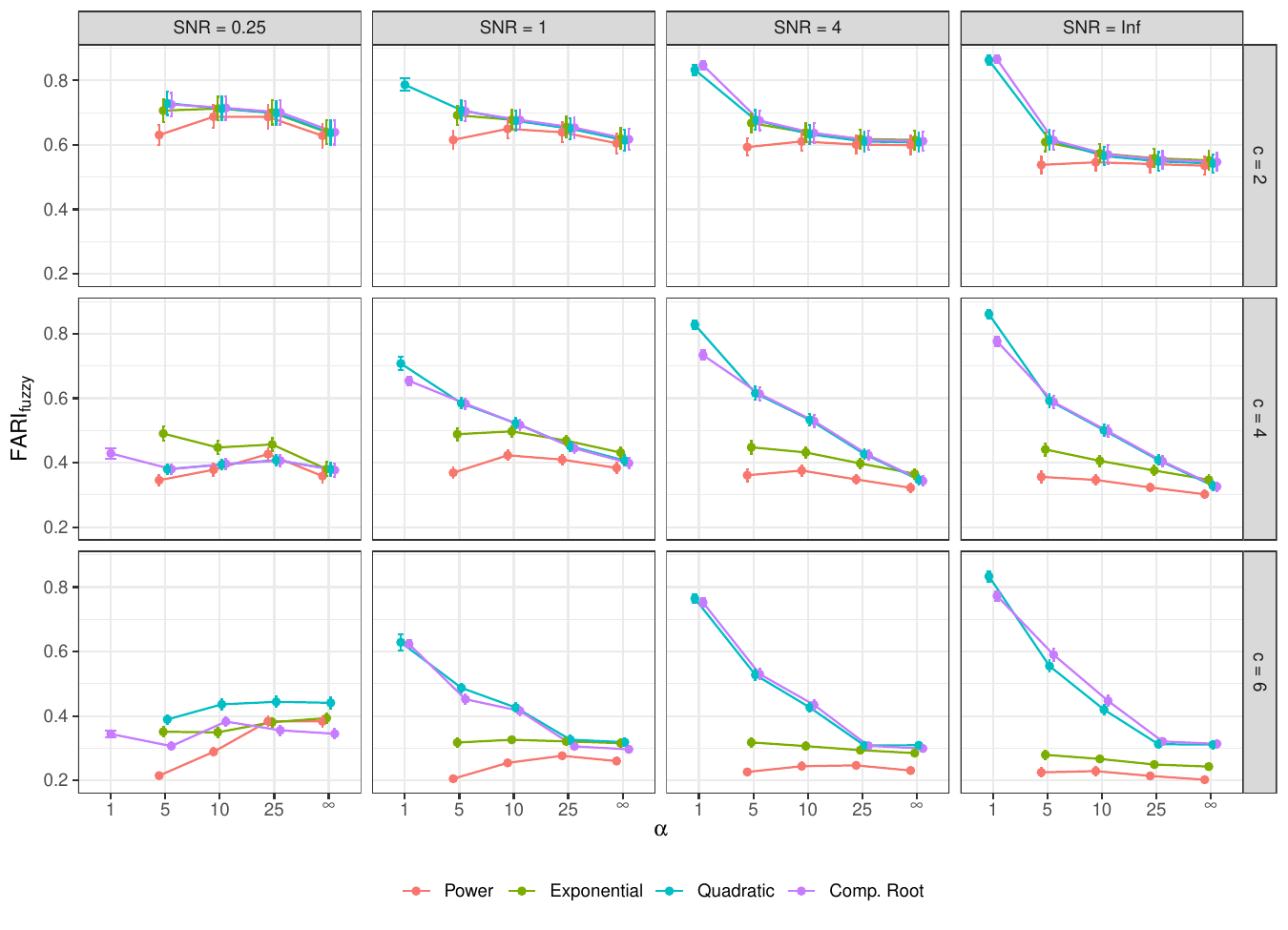}
    \caption{Mean $\mathrm{FARI}_\mathrm{fuzzy}$ of KFRC against $\alpha$ per fuzzifier, by $c$ and SNR, with 95\% confidence intervals. Results are averaged over multi-cluster criteria and kernel combinations (sum and product).}
    \label{fig:alpha_sensitivity}
\end{figure}

\textcolor{black}{Figure~\ref{fig:alpha_sensitivity} shows the effect of $\alpha$ on $\mathrm{FARI}_\mathrm{fuzzy}$ in more detail. For SNR $\geq 1$, performance generally decreases as $\alpha$ increases, particularly for $c \geq 4$, whereas larger values of $\alpha$ are mainly beneficial when SNR $= 0.25$ and $c \geq 4$. This suggests running Stage 1 over an increasing grid of values of $\alpha$ and stopping at the first one for which a feasible solution is found. This yields a $\mathrm{FARI}_\mathrm{fuzzy}$ within 0.065 of that of the best $\alpha$ in all settings with SNR $\geq 1$, but performs worse when SNR $= 0.25$ and $c \geq 4$, where only few variables are informative. This is expected, since $\alpha$ bounds the effective number of variables from below, that is
\begin{equation*}
W_\mathrm{eff} \geq \min_{0 \leq k \le p} \frac{(k\alpha + p - k)^2}{k\alpha^2 + p - k}
\end{equation*}
(see Theorem~\ref{thm:weff_alpha} in \ref{appen5}), so small values of $\alpha$ prevent the kernel from concentrating on few variables. For instance, with $p = 10$ the bound equals approximately 5.6, 3.3 and 1.8 for $\alpha = 5, 10$ and $25$, so concentrating on the two informative variables when SNR $= 0.25$ requires $\alpha \approx 25$. Examples of lower bounds on $W_\text{eff}$ for a range of $p$ and $\alpha$ values are illustrated in Figure~\ref{fig:weff_bounds} in \ref{appen5}.}

\begin{figure}[!ht]
    \centering
    \includegraphics[width=\linewidth]{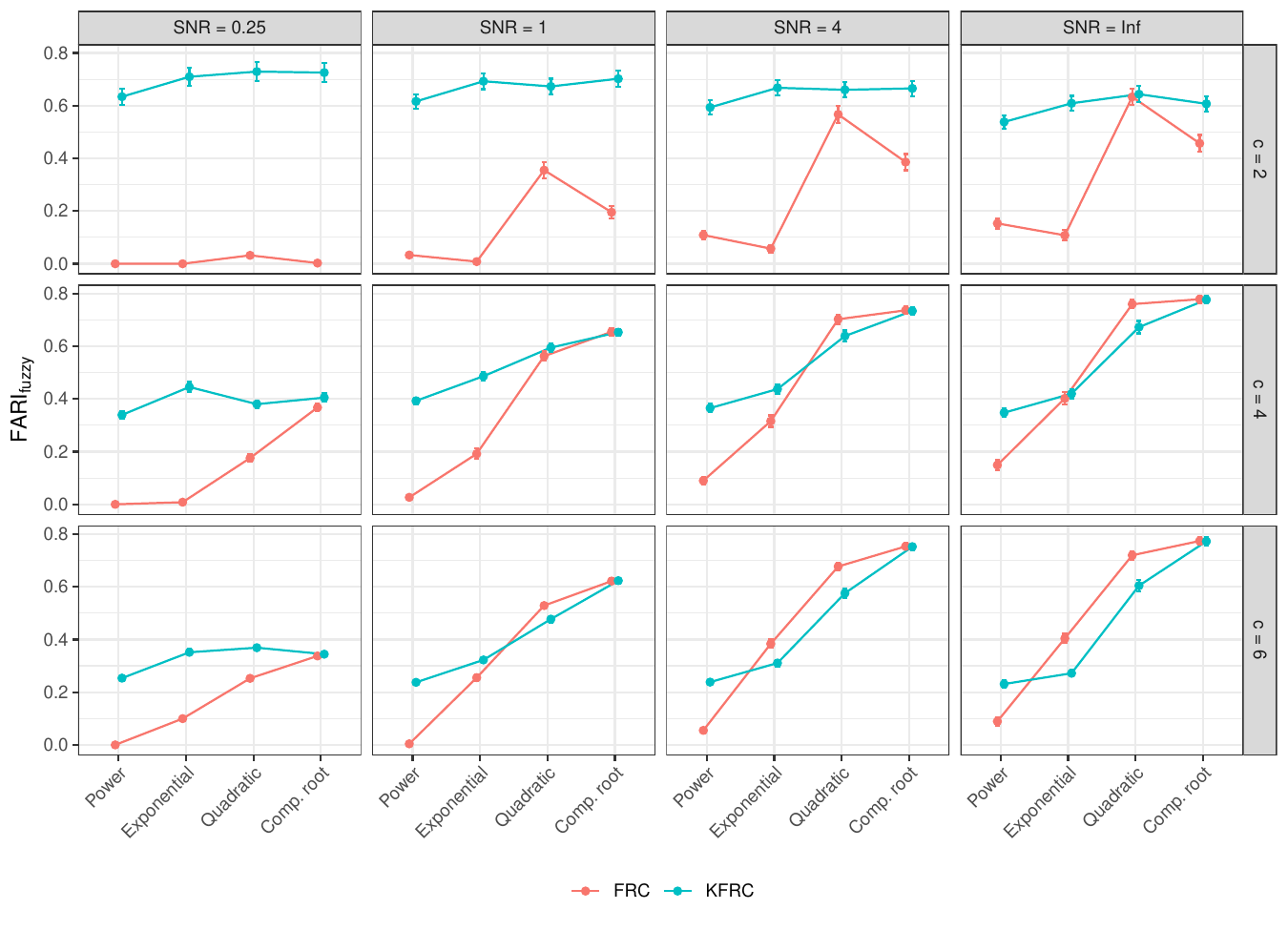}
    \caption{Mean $\mathrm{FARI}_\mathrm{fuzzy}$ of FRC and KFRC per fuzzifier, by $c$ and SNR, with 95\% confidence intervals. KFRC uses the smallest feasible $\alpha$.}
    \label{fig:ablation_sims}
\end{figure}

\textcolor{black}{In order to separate the contributions of the bandwidth selection procedure and the fuzzifier function, we fit a two-way repeated-measures ANOVA, treating the method (FRC vs. KFRC) and fuzzifier function as factors. We use KFRC with the smallest feasible anisotropy bound $\alpha$ and the resulting interaction plots are shown in Figure~\ref{fig:ablation_sims}. The interaction between the two factors is large in all settings with the partial $\eta^2$ values ranging between 0.43 and 0.77 (see Table~\ref{tab:ablation_eta} in \ref{appen3}). Bandwidth selection substantially improves the power fuzzifier in all settings, the remaining fuzzifiers in most settings with $c = 2$, and the exponential and quadratic fuzzifiers when SNR $= 0.25$. The effect is small for the complementary root fuzzifier when $c \geq 4$, and performance is reduced for the quadratic fuzzifier when SNR $\geq 4$ and $c \geq 4$, and for the exponential fuzzifier when SNR $\geq 4$ and $c = 6$.}

\subsection{Benchmark data sets}\label{subsec:benchmarkdata}

We perform fuzzy clustering on six publicly available benchmark data sets using the same configurations as the ones in Section~\ref{subsec:syntheticdata}. Our results are summarised in Table~\ref{tab:realdata_res}. Notice that for KFRC, FRC, and FANNY, the results reported correspond to the configurations for which the highest $\mathrm{FARI}_\mathrm{crisp}$ was obtained. \textcolor{black}{We remark that the bandwidth selection procedure was done in a fully unsupervised manner; class labels were used only after clustering, to select which configuration to report in Table~\ref{tab:realdata_res}. Results for all configurations are shown in Figure~\ref{fig:realdata_bench}.} KFRC is characterised by a strong overall classification performance, achieving equal or even higher ARI and AMI scores than most of its competitors. This is particularly the case on the Rice, Seeds, and Image Segmentation data sets. Moreover, its uniformity index lies at intermediate values, avoiding both the near-crisp partitions of vFCM and the near-uniform partitions of FANNY and FCM-MK.

\begin{table*}[!ht]
\centering
\small
\setlength{\tabcolsep}{3pt}
\caption{Clustering performance of fuzzy algorithms on six benchmark datasets. Best scores per (dataset, metric) appear in bold.}
\begin{tabular}{llcccccc}
\toprule
Dataset &  & KFRC & FRC & FANNY & MSFCM & vFCM & FCM-MK \\
\midrule
\multirow{5}{*}{\shortstack[l]{Diabetes\\($n=768$,\\$p=8$, $c=2$)}} & ARI & 0.164 & 0.164 & 0.164 & \textbf{0.181} & 0.116 & 0.172 \\
 & AMI & 0.105 & 0.106 & 0.125 & \textbf{0.130} & 0.062 & 0.123 \\
 & $\mathrm{FARI}_{\mathrm{crisp}}$ & 0.096 & 0.098 & 0.000 & 0.014 & \textbf{0.115} & 0.008 \\
 & $U$ & 0.588 & 0.582 & 1.000 & 0.989 & 0.011 & 0.996 \\
\midrule
\multirow{5}{*}{\shortstack[l]{Rice\\($n=3810$,\\$p=7$, $c=2$)}} & ARI & \textbf{0.732} & 0.681 & 0.669 & 0.682 & 0.681 & 0.684 \\
 & AMI & \textbf{0.622} & 0.566 & 0.555 & 0.568 & 0.566 & 0.570 \\
 & $\mathrm{FARI}_{\mathrm{crisp}}$ & \textbf{0.712} & 0.628 & 0.248 & 0.405 & 0.682 & 0.391 \\
 & $U$ & 0.069 & 0.144 & 0.819 & 0.542 & 0.004 & 0.563 \\
\midrule
\multirow{5}{*}{\shortstack[l]{Seeds\\($n=210$,\\$p=7$, $c=3$)}} & ARI & \textbf{0.785} & \textbf{0.785} & 0.705 & \textbf{0.785} & \textbf{0.785} & 0.761 \\
 & AMI & \textbf{0.736} & \textbf{0.736} & 0.666 & \textbf{0.736} & \textbf{0.736} & 0.714 \\
 & $\mathrm{FARI}_{\mathrm{crisp}}$ & 0.691 & 0.692 & 0.282 & 0.580 & \textbf{0.784} & 0.451 \\
 & $U$ & 0.148 & 0.147 & 0.798 & 0.375 & 0.001 & 0.553 \\
\midrule
\multirow{5}{*}{\shortstack[l]{Vehicle\\($n=846$,\\$p=18$, $c=4$)}} & ARI & \textbf{0.135} & 0.069 & 0.080 & 0.067 & 0.066 & 0.072 \\
 & AMI & \textbf{0.161} & 0.087 & 0.099 & 0.083 & 0.087 & 0.084 \\
 & $\mathrm{FARI}_{\mathrm{crisp}}$ & \textbf{0.149} & 0.060 & 0.024 & 0.052 & 0.066 & 0.040 \\
 & $U$ & 0.060 & 0.361 & 0.947 & 0.613 & 0.002 & 0.801 \\
\midrule
\multirow{5}{*}{\shortstack[l]{Wine\\($n=178$,\\$p=13$, $c=3$)}} & ARI & 0.897 & 0.897 & 0.368 & \textbf{0.915} & 0.897 & 0.897 \\
 & AMI & 0.872 & 0.872 & 0.385 & \textbf{0.889} & 0.872 & 0.872 \\
 & $\mathrm{FARI}_{\mathrm{crisp}}$ & 0.676 & 0.678 & 0.000 & 0.429 & \textbf{0.898} & 0.316 \\
 & $U$ & 0.232 & 0.229 & 1.000 & 0.692 & 0.000 & 0.801 \\
\midrule
\multirow{5}{*}{\shortstack[l]{Image Segmentation\\($n=2310$,\\$p=18$, $c=7$)}} & ARI & \textbf{0.620} & 0.539 & 0.330 & 0.483 & 0.528 & 0.539 \\
 & AMI & \textbf{0.699} & 0.626 & 0.431 & 0.565 & 0.607 & 0.632 \\
 & $\mathrm{FARI}_{\mathrm{crisp}}$ & 0.504 & 0.430 & 0.077 & 0.365 & \textbf{0.529} & 0.235 \\
 & $U$ & 0.282 & 0.315 & 0.958 & 0.610 & 0.005 & 0.792 \\
\bottomrule
\end{tabular}
\label{tab:realdata_res}
\end{table*}

Some of the ARI and AMI scores are low for all data sets. This is because while the choice of classification data sets for benchmarking cluster analysis is common practice, classification labels need not agree with the clustering objectives \cite{hennig2015true}. What matters in unsupervised clustering is the interpretation of the clusters, whereas the groups induced by the classification labels need not be optimal in any sense \cite{bautista2025ground}. We investigate the obtained KFRC partitions by projecting the data onto the unit sphere via kernel PCA, weighting the colour of each observation by its fuzzy membership degree. Figure~\ref{fig:realdata_kpca_plots} reveals that the clusters exhibit clear spatial separation in this induced space. The Vehicle data set provides a particularly interesting example. Although KFRC achieved a relatively low ARI score, the projection reveals that the algorithm successfully concentrated the points into four distinct coherent regions, something that is further justified by the near-zero uniformity index value.

\begin{figure}[!ht]
    \centering
    \includegraphics[width=\linewidth]{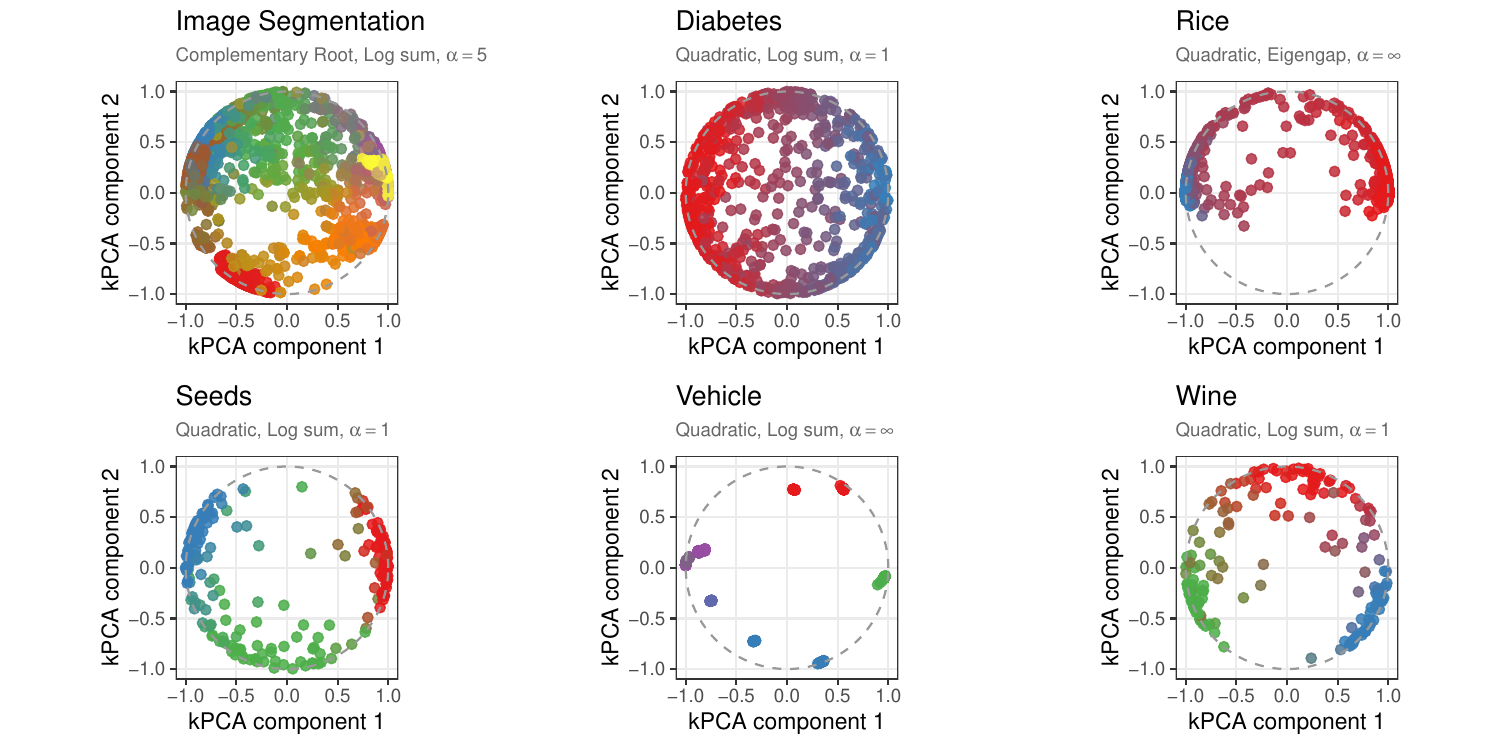}
    \caption{KFRC fuzzy memberships projected onto the unit sphere in feature space, with colours blending cluster prototypes by membership probability.}
    \label{fig:realdata_kpca_plots}
\end{figure}

\textcolor{black}{Based on Figure~\ref{fig:realdata_bench} we see that the choice of fuzzifier has a substantial effect on the clustering output, with the quadratic and complementary root fuzzifiers outperforming the power and exponential ones on all data sets except Diabetes, where all fuzzifiers perform quite poorly. Moreover, within a fixed fuzzifier function, most of the variation is due to the anisotropy bound $\alpha$, while the multi-cluster criterion has a comparatively small effect. The direction of the effect of $\alpha$ is data-dependent: larger values improve FARI$_{\text{crisp}}$ on Rice and Vehicle but reduce it on Seeds and Wine, suggesting that most features are informative in the latter. Finally, KFRC outperforms FANNY and FCM-MK in the majority of configurations, while vFCM attains a higher FARI$_{\text{crisp}}$ on several data sets at the cost of near-crisp partitions (see Table~\ref{tab:realdata_res}).}

\begin{figure}[!ht]
    \centering
    \includegraphics[width=\linewidth]{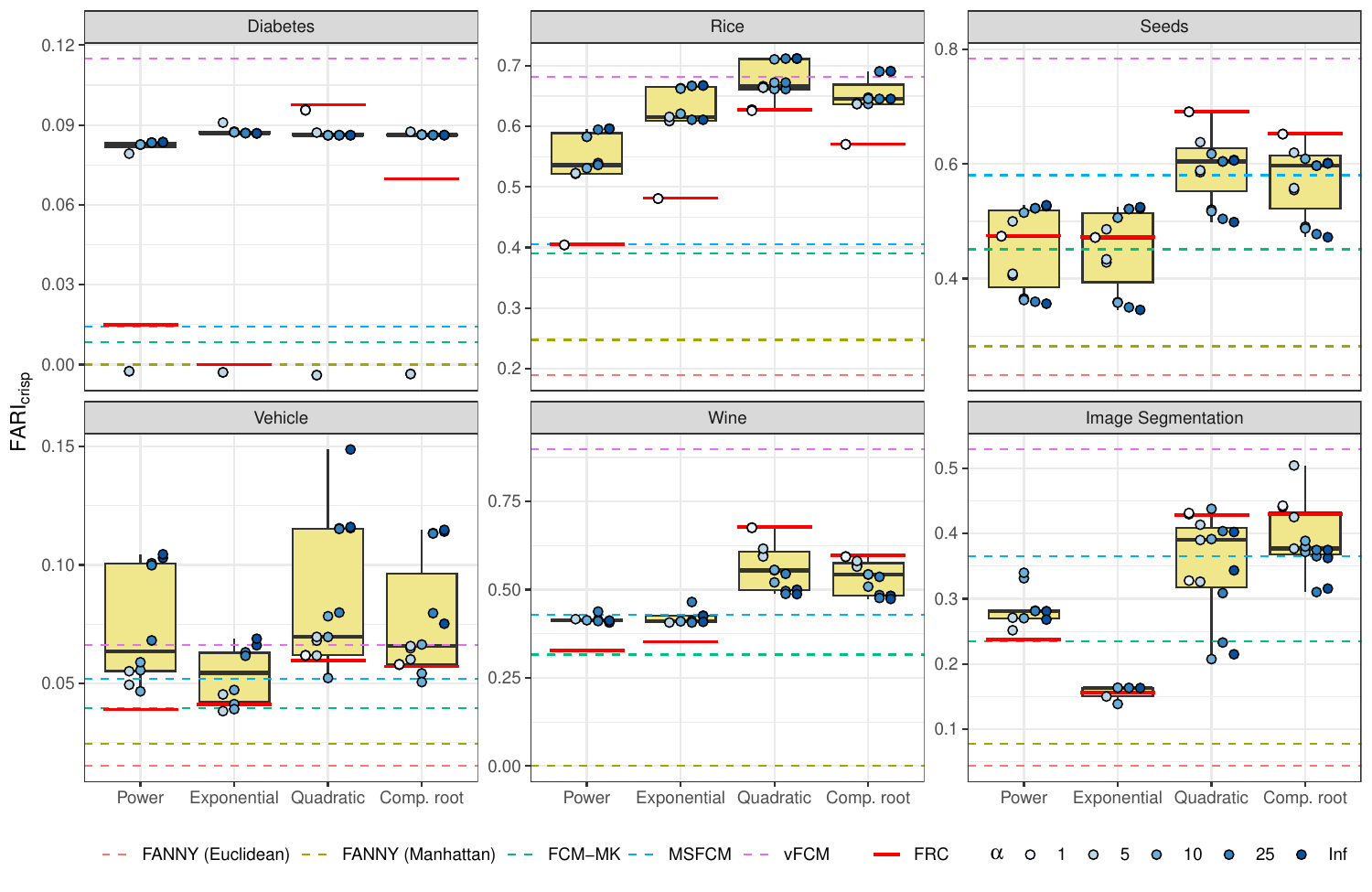}
    \caption{FARI$_{\text{crisp}}$ of KFRC on the six benchmark data sets across all bandwidth selection configurations (sum kernel). Each box summarises the configurations for a given fuzzifier (three multi-cluster criteria $\times$ five anisotropy bounds $\alpha$), with points shaded by $\alpha$. Solid red lines indicate the performance of FRC with the same fuzzifier, and the dashed lines show the remaining competitors.}
    \label{fig:realdata_bench}
\end{figure}

\section{Conclusion}\label{sec:conclusion}

In this article we discussed the selection of kernel bandwidths for kernel fuzzy relational clustering. We argued that defining dissimilarities via kernel functions allows controlling the data geometry, thus enabling variable weighting and/or selection. Moreover, we obtained the conditions under which the choice of the fuzzifier parameter controls the fuzziness of the solution without the risk of collapsing to uniform membership probabilities, even for large fuzzifier values. As a result, we derived a novel fuzzifier function that yields a relaxed stability criterion and we devised a two-stage bandwidth selection algorithm that leverages stability and clustering quality.

\textcolor{black}{Our simulations on both synthetic and publicly available data sets indicate that KFRC improves cluster recovery over existing fuzzy relational clustering algorithms in many settings, particularly when the cluster structure is defined in a subspace of the feature space, while avoiding both uniform and near-crisp solutions. Our analysis further shows that the two proposed components are complementary. Specifically, the complementary root fuzzifier alone avoids uniform collapse when the signal is moderate to strong and the number of clusters is large, whereas the bandwidth selection procedure is most beneficial when the signal is weak, when $c = 2$, or when using fuzzifier functions that collapse to the uniform solution without it, such as the power and exponential fuzzifiers.}

This work presents several limitations that may be addressed in future studies. Specifically, the value of the anisotropy bound is a critical choice in the success of the algorithm. \textcolor{black}{While $\alpha$ can be interpreted through the lower bound it implies on the effective number of variables, and selecting the smallest feasible $\alpha$ performs well in most settings, this default is less effective when only few variables are informative and the number of clusters is large.} Future work could explore the selection of $\alpha$ in more theoretical detail, potentially even devising a dynamic updating scheme for $\alpha$. Finally, the use of gradient ascent may lead to slow convergence for large sample sizes, hence alternative optimisation schemes may need to be employed.

\newpage
\appendix
\renewcommand{\thetheorem}{\Alph{section}.\arabic{theorem}}
\counterwithin*{algorithm}{section}
\renewcommand{\thealgorithm}{\Alph{section}.\arabic{algorithm}}
\counterwithin*{table}{section}
\renewcommand{\thetable}{\Alph{section}.\arabic{table}}
\section{Analytic derivative expressions for chain rule}\label{appen1}
\vspace{-1cm}
\begin{align*}
    \frac{\partial\mathbf{F}}{\partial \bar{\mathbf{K}}(\boldsymbol{\theta})} & = \frac{1}{n}\mathbf{D}(\boldsymbol{\theta})^{-1/2} \otimes \mathbf{D}(\boldsymbol{\theta})^{-1/2}, \\
    \frac{\partial \bar{\mathbf{K}}(\boldsymbol{\theta})}{\partial \theta_j} & = \mathbf{H}\frac{\partial \mathbf{K}(\boldsymbol{\theta})}{\partial \theta_j}\mathbf{H},\\
    \frac{\partial \mathbf{F}}{\partial \mathbf{D}(\boldsymbol{\theta})}& = -\frac{1}{2n}\left(\mathbf{D}(\boldsymbol{\theta})^{-3/2}\bar{\mathbf{K}}(\boldsymbol{\theta})\mathbf{D}(\boldsymbol{\theta})^{-1/2} +  \mathbf{D}(\boldsymbol{\theta})^{-1/2}\bar{\mathbf{K}}(\boldsymbol{\theta})\mathbf{D}(\boldsymbol{\theta})^{-3/2} \right),\\
    \frac{\partial \mathbf{D}(\boldsymbol{\theta})}{\partial \theta_j} & = \mathrm{diag}\left( \frac{\partial \bar{\mathbf{K}}(\boldsymbol{\theta})}{\partial \theta_j} \right) = \mathrm{diag}\left(\mathbf{H}\left( \frac{\partial \mathbf{K}(\boldsymbol{\theta})}{\partial \theta_j}\right) \mathbf{H}\right).
\end{align*}

\section{Algorithms}\label{appen2}
\begin{algorithm}[H]
\scriptsize
\caption{Anisotropy Projection (\textproc{ProjAniso})}
\label{alg:proj_aniso}
\begin{algorithmic}[1]
\Require Log-bandwidth vector $\ell(\boldsymbol{\theta}) \in \mathbb{R}^p$, bound $\log\alpha \geq 0$.
\Ensure $\ell(\boldsymbol{\theta})_{\text{proj}}$ with $\max(\ell(\boldsymbol{\theta})_{\text{proj}}) - \min(\ell(\boldsymbol{\theta})_{\text{proj}}) \le \log\alpha$.
\If{$\max(\ell(\boldsymbol{\theta})) - \min(\ell(\boldsymbol{\theta})) \le \log\alpha$}
    \State \Return $\ell(\boldsymbol{\theta})$
\Else
    \State $\bar{\eta} \gets \frac{1}{p}\sum_{j=1}^p \eta_j$
    \State $\ell(\boldsymbol{\theta})_{\text{proj}} \gets \bar{\eta} + \frac{\log\alpha}{\max(\ell(\boldsymbol{\theta}))-\min(\ell(\boldsymbol{\theta}))} (\ell(\boldsymbol{\theta}) - \bar{\eta})$
    \State \Return $\ell(\boldsymbol{\theta})_{\text{proj}}$
\EndIf
\end{algorithmic}
\end{algorithm}

\begin{algorithm}[H]
\scriptsize
\caption{Feasibility Projection (\textproc{ProjFeas})}
\label{alg:proj_feas}
\begin{algorithmic}[1]
\Require Data $\mathbf{X}$, current log-bandwidth $\ell(\boldsymbol{\theta})$, proposed update $\ell(\boldsymbol{\theta})_{\text{new}}$, stability threshold $\tau > 0$, symmetric PSD kernel function $k$.
\Statex \textit{Internal constants:} Line-search tolerance $\delta = 10^{-6}$, minimum step factor $\beta_{\min} = 10^{-4}$.
\Ensure Log-bandwidths $\ell(\boldsymbol{\theta})^*$ with $\lambda_1(\exp(\ell(\boldsymbol{\theta})^*)) \ge \tau - \delta$.
\State $\beta \gets 1$
\While{$\beta > \beta_{\min}$}
    \State $\ell(\boldsymbol{\theta})_{\text{trial}} \gets \ell(\boldsymbol{\theta}) + \beta(\ell(\boldsymbol{\theta})_{\text{new}} - \ell(\boldsymbol{\theta}))$
    \State $\boldsymbol{\theta}_{\text{trial}} \gets \exp(\ell(\boldsymbol{\theta})_{\text{trial}})$
    \State Compute $\mathbf{F}(\boldsymbol{\theta}_{\text{trial}})$ using $k$ and extract $\lambda_1 \equiv \lambda_1(\mathbf{F})$.
    \If{$\lambda_1 \ge \tau - \delta$}
        \State \Return $\ell(\boldsymbol{\theta})_{\text{trial}}$
    \EndIf
    \State $\beta \gets \beta / 2$
\EndWhile
\State \Return $\ell(\boldsymbol{\theta})$ 
\end{algorithmic}
\end{algorithm}

\begin{algorithm}[H]
\scriptsize
\caption{Feasibility; maximise $\lambda_1(\mathbf{F})$ (\textproc{Stage1})}
\label{alg:stage1}
\begin{algorithmic}[1]
\Require Data $\mathbf{X} \in \mathbb{R}^{n \times p}$, anisotropy bound $\alpha \ge 1$, number of random starts $S$, learning rate $\eta$, momentum $\gamma$, tolerance $\varepsilon$, max iterations $T$, stability threshold $\tau > 0$, symmetric PSD kernel function $k$.
\Ensure Set $\mathcal{F}$ of feasible pairs $(\ell(\boldsymbol{\theta}), \mathbf{F})$ with $\lambda_1(\exp(\ell(\boldsymbol{\theta}))) \ge \tau$, or $\mathcal{F} = \varnothing$ if infeasible.
\State $\mathcal{F} \gets \varnothing$
\State Compute column standard deviations $\widehat{\sigma}_1,\ldots,\widehat{\sigma}_p$ of $\mathbf{X}$.
\For{$s = 1$ to $S$}
    \If{$s = 1$}
        \State $\ell(\boldsymbol{\theta}) \gets \log \widehat{\boldsymbol{\sigma}}$
    \Else
        \State $\ell(\boldsymbol{\theta}) \gets \log \widehat{\boldsymbol{\sigma}} + \mathcal{N}(0,0.5^2\mathbf{I}_p)$
    \EndIf
    \If{$\alpha < \infty$}
        \State $\ell(\boldsymbol{\theta}) \gets \Call{ProjAniso}{\ell(\boldsymbol{\theta}), \log\alpha}$
    \EndIf
    
    \State $\mathbf{m} \gets \mathbf{0}$
    \For{$t = 1$ to $T$}
        \State $\boldsymbol{\theta} \gets \exp(\ell(\boldsymbol{\theta}))$
        \State Construct $\mathbf{F}(\boldsymbol{\theta})$ using $k$, extract $\lambda_1(\mathbf{F})$ and corresponding eigenvector $\mathbf{v}_1$.
        \State $\nabla_{\ell(\boldsymbol{\theta})} \lambda_1 \gets \boldsymbol{\theta} \odot \bigl( \mathbf{v}_1^\top \frac{\partial \mathbf{F}}{\partial \boldsymbol{\theta}} \mathbf{v}_1 \bigr)$.
        \State $\mathbf{m} \gets \gamma \mathbf{m} + \eta \nabla_{\ell(\boldsymbol{\theta})} \lambda_1$
        \State $\ell(\boldsymbol{\theta})_{\text{new}} \gets \ell(\boldsymbol{\theta}) + \mathbf{m}$
        \If{$\alpha < \infty$}
            \State $\ell(\boldsymbol{\theta})_{\text{new}} \gets \Call{ProjAniso}{\ell(\boldsymbol{\theta})_{\text{new}}, \log\alpha}$
        \EndIf
        \If{$\|\ell(\boldsymbol{\theta})_{\text{new}} - \ell(\boldsymbol{\theta})\|_\infty < \varepsilon$}
            \State $\ell(\boldsymbol{\theta}) \gets \ell(\boldsymbol{\theta})_{\text{new}}$
            \State \textbf{break}
        \EndIf
        \State $\ell(\boldsymbol{\theta}) \gets \ell(\boldsymbol{\theta})_{\text{new}}$
    \EndFor
    
    \State $\boldsymbol{\theta} \gets \exp(\ell(\boldsymbol{\theta}))$
    \State Compute $\mathbf{F}(\boldsymbol{\theta})$ using $k$ and extract $\lambda_1(\mathbf{F})$.
    \If{$\lambda_1 \ge \tau$}
        \State Store the pair $(\ell(\boldsymbol{\theta}), \mathbf{F})$ in $\mathcal{F}$.
    \EndIf
\EndFor
\State \Return $\mathcal{F}$
\end{algorithmic}
\end{algorithm}

\begin{algorithm}[H]
\scriptsize
\caption{Optimise Multi-Cluster Criterion (\textproc{Stage2})}
\label{alg:stage2}
\begin{algorithmic}[1]
\Require Data $\mathbf{X}$, number of clusters $c$, anisotropy bound $\alpha$, one feasible candidate $(\ell(\boldsymbol{\theta})_0, \mathbf{F}_0)$ from Stage 1, learning rate $\eta$, momentum $\gamma$, tolerance $\varepsilon$, max iterations $T$, stability threshold $\tau > 0$, symmetric PSD kernel function $k$, multi-cluster objective $\mathcal{Q}$.
\Ensure Optimised bandwidth vector $\boldsymbol{\theta}^*$, maximised multi-cluster criterion $\mathcal{Q}$, and eigenvalue summary.
\State $\ell(\boldsymbol{\theta}) \gets \ell(\boldsymbol{\theta})_0$.
\State $\mathbf{m} \gets \mathbf{0}$
\For{$t = 1$ to $T$}
    \State $\boldsymbol{\theta} \gets \exp(\ell(\boldsymbol{\theta}))$
    \State $\nabla_{\ell(\boldsymbol{\theta})} \lambda_i \gets \boldsymbol{\theta} \odot \bigl( \mathbf{v}_i^\top \frac{\partial \mathbf{F}}{\partial \boldsymbol{\theta}} \mathbf{v}_i \bigr)$.
    \State $\mathbf{m} \gets \gamma \mathbf{m} + \frac{\eta}{2} \nabla_{\ell(\boldsymbol{\theta})}\mathcal{Q}$
    \State $\ell(\boldsymbol{\theta})_{\text{new}} \gets \ell(\boldsymbol{\theta}) + \mathbf{m}$
    \If{$\alpha < \infty$}
        \State $\ell(\boldsymbol{\theta})_{\text{new}} \gets \Call{ProjAniso}{\ell(\boldsymbol{\theta})_{\text{new}}, \log\alpha}$
    \EndIf
    \State $\ell(\boldsymbol{\theta})_{\text{new}} \gets \Call{ProjFeas}{\mathbf{X}, \ell(\boldsymbol{\theta}), \ell(\boldsymbol{\theta})_{\text{new}}, \tau, k}$
    \If{$\|\ell(\boldsymbol{\theta})_{\text{new}} - \ell(\boldsymbol{\theta})\|_\infty < \varepsilon$}
        \State $\ell(\boldsymbol{\theta}) \gets \ell(\boldsymbol{\theta})_{\text{new}}$
        \State \textbf{break}
    \Else
        \State $\ell(\boldsymbol{\theta}) \gets \ell(\boldsymbol{\theta})_{\text{new}}$
        \State Re-compute $\mathbf{F}(\exp(\ell(\boldsymbol{\theta})))$ using $k$.
    \EndIf
\EndFor
\State $\boldsymbol{\theta}^* \gets \exp(\ell(\boldsymbol{\theta}))$
\State Compute $\mathbf{F}(\boldsymbol{\theta}^*)$ using $k$, the necessary leading eigenpairs, and update $\mathcal{Q}$.
\State \Return $\boldsymbol{\theta}^*$, $\mathcal{Q}$, $\lambda_1$, relevant leading eigenpairs.
\end{algorithmic}
\end{algorithm}

\section{Synthetic simulation results}\label{appen3}
{\scriptsize
\setlength{\tabcolsep}{2pt}
\begin{longtable}[l]{llcccc}
\caption{Mean cluster recovery and uniformity across $c = 2/4/6$ for each method and signal-to-noise ratio. Best KFRC results over anisotropy bound and criterion per dataset are reported. Best scores for each (SNR, $c$) combination appear in bold.}
\label{tab:simresults} \\
\toprule
Method &  &  SNR = 0.25 & SNR = 1 & SNR = 4 & SNR = $\infty$ \\
\midrule
\endfirsthead

\multicolumn{6}{l}{{\tablename\ \thetable{} - continued from previous page}} \\
\toprule
Method &  &  SNR = 0.25 & SNR = 1 & SNR = 4 & SNR = $\infty$ \\
\midrule
\endhead

\midrule
\multicolumn{6}{r}{{Continued on next page...}} \\
\endfoot

\bottomrule
\endlastfoot

\multirow{5}{*}{KFRC (Power)} & ARI & 0.701/0.570/0.518 & 0.641/0.509/0.283 & 0.613/0.466/0.278 & 0.553/0.436/0.306 \\
 & AMI & 0.647/0.622/0.628 & 0.572/0.549/0.395 & 0.539/0.494/0.367 & 0.477/0.467/0.377 \\
 & $\mathrm{FARI}_\mathrm{crisp}$ & 0.586/0.451/0.428 & 0.524/0.336/0.211 & 0.491/0.283/0.174 & 0.429/0.256/0.169 \\
 & $\mathrm{FARI}_\mathrm{fuzzy}$ & 0.751/0.676/0.650 & 0.696/0.552/0.376 & 0.663/0.458/0.313 & 0.595/0.420/0.282 \\
 & $U$ & 0.257/0.405/0.396 & 0.296/0.558/0.531 & 0.320/0.607/0.603 & 0.363/0.618/0.658 \\
\midrule
\multirow{5}{*}{KFRC (Exponential)} & ARI & 0.704/0.562/0.506 & 0.657/0.515/0.299 & 0.639/0.460/0.296 & 0.582/0.426/0.286 \\
 & AMI & 0.649/0.615/0.610 & 0.589/0.556/0.401 & 0.568/0.494/0.378 & 0.509/0.463/0.358 \\
 & $\mathrm{FARI}_\mathrm{crisp}$ & 0.645/0.392/0.312 & 0.574/0.316/0.179 & 0.545/0.269/0.155 & 0.489/0.244/0.140 \\
 & $\mathrm{FARI}_\mathrm{fuzzy}$ & 0.768/0.699/0.608 & 0.726/0.625/0.415 & 0.706/0.539/0.367 & 0.648/0.506/0.329 \\
 & $U$ & 0.151/0.482/0.607 & 0.208/0.561/0.699 & 0.230/0.592/0.748 & 0.253/0.606/0.757 \\
\midrule
\multirow{5}{*}{KFRC (Quadratic)} & ARI & \textbf{0.711}/\textbf{0.617}/\textbf{0.588} & 0.671/0.624/0.489 & 0.682/0.619/0.557 & 0.657/0.622/0.578 \\
 & AMI & \textbf{0.657}/\textbf{0.658}/0.647 & 0.606/0.638/0.559 & 0.619/0.623/0.608 & 0.602/0.629/0.617 \\
 & $\mathrm{FARI}_\mathrm{crisp}$ & \textbf{0.685}/\textbf{0.541}/0.438 & 0.640/0.534/0.330 & 0.617/0.519/0.389 & 0.566/0.501/0.427 \\
 & $\mathrm{FARI}_\mathrm{fuzzy}$ & \textbf{0.775}/\textbf{0.755}/0.709 & 0.738/0.756/0.616 & 0.730/0.732/0.674 & 0.690/0.732/0.687 \\
 & $U$ & 0.081/0.256/0.430 & 0.103/0.268/0.509 & 0.135/0.273/0.462 & 0.166/0.309/0.413 \\
\midrule
\multirow{5}{*}{KFRC (Comp. root)} & ARI & 0.707/0.609/0.559 & 0.664/\textbf{0.644}/\textbf{0.535} & 0.646/0.656/0.628 & 0.594/0.683/0.657 \\
 & AMI & 0.653/0.655/\textbf{0.657} & 0.598/\textbf{0.662}/\textbf{0.608} & 0.577/0.666/0.680 & 0.526/0.688/0.695 \\
 & $\mathrm{FARI}_\mathrm{crisp}$ & 0.656/0.527/\textbf{0.467} & 0.601/0.488/0.370 & 0.579/0.475/0.435 & 0.526/0.461/0.471 \\
 & $\mathrm{FARI}_\mathrm{fuzzy}$ & 0.772/0.754/\textbf{0.727} & 0.733/\textbf{0.780}/\textbf{0.682} & 0.714/\textbf{0.777}/\textbf{0.764} & 0.657/\textbf{0.786}/\textbf{0.783} \\
 & $U$ & 0.138/0.269/0.337 & 0.175/0.357/0.450 & 0.186/0.376/0.425 & 0.203/0.414/0.385 \\
\midrule
\multirow{5}{*}{FRC (Power)} & ARI & 0.488/0.277/0.232 & 0.648/0.426/0.294 & 0.681/0.529/0.373 & 0.698/0.570/0.401 \\
 & AMI & 0.443/0.266/0.244 & 0.593/0.388/0.289 & 0.628/0.493/0.373 & 0.648/0.542/0.400 \\
 & $\mathrm{FARI}_\mathrm{crisp}$ & 0.000/0.000/0.000 & 0.044/0.020/0.004 & 0.115/0.083/0.049 & 0.149/0.129/0.083 \\
 & $\mathrm{FARI}_\mathrm{fuzzy}$ & 0.000/0.000/0.000 & 0.033/0.026/0.004 & 0.109/0.090/0.055 & 0.153/0.149/0.089 \\
 & $U$ & 1.000/1.000/1.000 & 0.981/0.990/0.998 & 0.938/0.956/0.977 & 0.913/0.928/0.957 \\
\midrule
\multirow{5}{*}{FRC (Exponential)} & ARI & 0.477/0.281/0.253 & 0.643/0.492/0.425 & 0.671/0.570/0.499 & 0.690/0.602/0.525 \\
 & AMI & 0.434/0.266/0.290 & 0.588/0.477/0.470 & 0.617/0.561/0.553 & 0.639/0.597/0.568 \\
 & $\mathrm{FARI}_\mathrm{crisp}$ & 0.000/0.004/0.041 & 0.010/0.089/0.106 & 0.057/0.156/0.172 & 0.097/0.205/0.191 \\
 & $\mathrm{FARI}_\mathrm{fuzzy}$ & 0.000/0.008/0.100 & 0.008/0.192/0.255 & 0.057/0.316/0.384 & 0.108/0.402/0.404 \\
 & $U$ & 1.000/0.997/0.955 & 0.995/0.910/0.880 & 0.966/0.834/0.811 & 0.933/0.784/0.785 \\
\midrule
\multirow{5}{*}{FRC (Quadratic)} & ARI & 0.496/0.291/0.232 & 0.666/0.583/0.492 & 0.712/0.656/0.615 & 0.733/0.676/0.641 \\
 & AMI & 0.450/0.320/0.329 & 0.611/0.613/0.569 & 0.657/0.673/0.668 & 0.682/0.688/0.683 \\
 & $\mathrm{FARI}_\mathrm{crisp}$ & 0.036/0.077/0.097 & 0.261/0.286/0.250 & 0.414/0.412/0.371 & 0.468/0.469/0.420 \\
 & $\mathrm{FARI}_\mathrm{fuzzy}$ & 0.032/0.175/0.252 & 0.355/0.562/0.529 & 0.567/0.702/0.677 & 0.633/0.759/0.720 \\
 & $U$ & 0.980/0.905/0.813 & 0.745/0.665/0.684 & 0.560/0.517/0.559 & 0.494/0.450/0.495 \\
\midrule
\multirow{5}{*}{FRC (Complementary root)} & ARI & 0.493/0.283/0.220 & 0.653/0.589/0.489 & 0.693/0.655/0.621 & 0.712/0.679/0.646 \\
 & AMI & 0.448/0.355/0.322 & 0.598/0.625/0.575 & 0.640/0.675/0.679 & 0.661/0.690/0.693 \\
 & $\mathrm{FARI}_\mathrm{crisp}$ & 0.005/0.130/0.138 & 0.163/0.282/0.295 & 0.294/0.385/0.419 & 0.345/0.438/0.458 \\
 & $\mathrm{FARI}_\mathrm{fuzzy}$ & 0.003/0.367/0.337 & 0.195/0.654/0.621 & 0.386/0.736/0.753 & 0.457/0.778/0.774 \\
 & $U$ & 0.998/0.690/0.588 & 0.879/0.586/0.523 & 0.748/0.488/0.442 & 0.694/0.444/0.399 \\
\midrule
\multirow{5}{*}{FANNY (Euclidean)} & ARI & 0.522/0.236/0.167 & 0.649/0.325/0.224 & 0.678/0.359/0.256 & 0.693/0.390/0.244 \\
 & AMI & 0.473/0.242/0.188 & 0.594/0.304/0.235 & 0.624/0.325/0.260 & 0.642/0.347/0.256 \\
 & $\mathrm{FARI}_\mathrm{crisp}$ & 0.000/0.000/0.000 & 0.000/0.000/0.000 & 0.000/0.002/0.000 & 0.001/0.004/0.002 \\
 & $\mathrm{FARI}_\mathrm{fuzzy}$ & -0.000/-0.000/-0.000 & -0.000/-0.000/-0.000 & 0.000/0.001/0.000 & 0.000/0.004/0.001 \\
 & $U$ & 1.000/1.000/1.000 & 1.000/1.000/1.000 & 1.000/0.999/1.000 & 1.000/0.998/0.999 \\
\midrule
\multirow{5}{*}{FANNY (Manhattan)} & ARI & 0.579/0.275/0.226 & 0.668/0.355/0.252 & 0.682/0.400/0.271 & 0.690/0.400/0.255 \\
 & AMI & 0.530/0.273/0.241 & 0.615/0.326/0.257 & 0.629/0.352/0.269 & 0.639/0.351/0.254 \\
 & $\mathrm{FARI}_\mathrm{crisp}$ & 0.000/0.000/0.000 & 0.000/0.001/0.000 & 0.003/0.006/0.002 & 0.010/0.007/0.003 \\
 & $\mathrm{FARI}_\mathrm{fuzzy}$ & -0.000/-0.000/-0.000 & 0.000/0.000/0.000 & 0.001/0.005/0.002 & 0.005/0.008/0.002 \\
 & $U$ & 1.000/1.000/1.000 & 1.000/1.000/1.000 & 0.999/0.998/0.999 & 0.997/0.997/0.999 \\
\midrule
\multirow{5}{*}{FCM-MK} & ARI & 0.484/0.272/0.234 & 0.639/0.422/0.294 & 0.673/0.521/0.371 & 0.691/0.564/0.392 \\
 & AMI & 0.441/0.266/0.250 & 0.586/0.383/0.292 & 0.621/0.485/0.371 & 0.644/0.536/0.392 \\
 & $\mathrm{FARI}_\mathrm{crisp}$ & 0.000/0.000/0.000 & 0.040/0.015/0.004 & 0.104/0.075/0.042 & 0.135/0.118/0.071 \\
 & $\mathrm{FARI}_\mathrm{fuzzy}$ & 0.000/0.000/0.000 & 0.029/0.018/0.004 & 0.095/0.079/0.047 & 0.135/0.130/0.075 \\
 & $U$ & 1.000/1.000/1.000 & 0.983/0.993/0.999 & 0.946/0.962/0.980 & 0.923/0.937/0.965 \\
\midrule
\multirow{5}{*}{vFCM} & ARI & 0.519/0.254/0.204 & \textbf{0.745}/0.591/0.487 & \textbf{0.792}/\textbf{0.694}/\textbf{0.640} & \textbf{0.805}/\textbf{0.712}/\textbf{0.681} \\
 & AMI & 0.478/0.266/0.246 & \textbf{0.688}/0.623/0.556 & \textbf{0.736}/\textbf{0.706}/\textbf{0.689} & \textbf{0.751}/\textbf{0.719}/\textbf{0.716} \\
 & $\mathrm{FARI}_\mathrm{crisp}$ & 0.335/0.088/0.083 & \textbf{0.714}/\textbf{0.575}/\textbf{0.465} & \textbf{0.789}/\textbf{0.692}/\textbf{0.638} & \textbf{0.803}/\textbf{0.712}/\textbf{0.681} \\
 & $\mathrm{FARI}_\mathrm{fuzzy}$ & 0.348/0.094/0.091 & \textbf{0.746}/0.610/0.500 & \textbf{0.826}/0.736/0.690 & \textbf{0.840}/0.761/0.740 \\
 & $U$ & 0.623/0.728/0.635 & 0.118/0.070/0.118 & 0.030/0.008/0.015 & 0.028/0.000/0.003 \\
\midrule
\multirow{5}{*}{MSFCM} & ARI & 0.491/0.281/0.233 & 0.645/0.474/0.321 & 0.682/0.595/0.454 & 0.700/0.627/0.491 \\
 & AMI & 0.446/0.266/0.244 & 0.590/0.455/0.322 & 0.629/0.577/0.456 & 0.649/0.615/0.489 \\
 & $\mathrm{FARI}_\mathrm{crisp}$ & 0.000/0.000/0.000 & 0.039/0.060/0.024 & 0.110/0.177/0.123 & 0.143/0.228/0.179 \\
 & $\mathrm{FARI}_\mathrm{fuzzy}$ & 0.000/0.000/0.000 & 0.028/0.065/0.022 & 0.102/0.196/0.135 & 0.144/0.264/0.195 \\
 & $U$ & 1.000/1.000/1.000 & 0.984/0.969/0.989 & 0.942/0.894/0.932 & 0.918/0.859/0.895 \\
\end{longtable}
}

\begin{table}[!ht]
\centering
\footnotesize
\caption{Partial $\eta^2$ from the two-way repeated-measures ANOVA on $\mathrm{FARI}_\mathrm{fuzzy}$, by $c$ and SNR.}
\label{tab:ablation_eta}
\begin{tabular}{rrccc}
\toprule
$c$ & SNR & Method & Fuzzifier & Method $\times$ fuzzifier \\
\midrule
2 & 0.25 & 0.798 & 0.559 & 0.434 \\
  & 1    & 0.813 & 0.574 & 0.461 \\
  & 4    & 0.766 & 0.715 & 0.679 \\
  & $\infty$ & 0.735 & 0.751 & 0.721 \\
\midrule
4 & 0.25 & 0.794 & 0.653 & 0.649 \\
  & 1    & 0.852 & 0.863 & 0.736 \\
  & 4    & 0.685 & 0.905 & 0.709 \\
  & $\infty$ & 0.302 & 0.906 & 0.592 \\
\midrule
6 & 0.25 & 0.811 & 0.774 & 0.674 \\
  & 1    & 0.668 & 0.935 & 0.770 \\
  & 4    & 0.003 & 0.933 & 0.734 \\
  & $\infty$ & 0.356 & 0.924 & 0.680 \\
\bottomrule
\end{tabular}
\end{table}

\section{Proof of conditions holding for the complementary root fuzzifier}\label{appen4}

\noindent Let the complementary root fuzzifier be defined as $t_\text{cr}(u) = 1 - (1-u)^{1/m}$, where the fuzziness parameter is $m > 1$ and the membership probability is $u \in [0,1]$.

\begin{theorem}
The complementary root fuzzifier $t_\text{cr}(u)$ satisfies all five conditions required for geometric fidelity and stable convergence in fuzzy relational clustering.
\end{theorem}

\begin{proof}
We evaluate $t_\text{cr}(u)$ against each of the five conditions sequentially. The first and second derivatives of $t_\text{cr}(u)$ with respect to $u$ are given by:
\begin{align}
    t'_{cr}(u) &= \frac{1}{m}(1-u)^{\frac{1-m}{m}}, \label{eq:deriv1} \\
    t''_{cr}(u) &= \frac{m-1}{m^2}(1-u)^{\frac{1-2m}{m}} \label{eq:deriv2}
\end{align}

\noindent \textbf{Condition 1: Boundary guarding.} \\
Evaluating the function at the boundaries $u=0$ and $u=1$:
\begin{align*}
    t_\text{cr}(0) &= 1 - (1-0)^{1/m} = 1 - 1 = 0, \\
    t_\text{cr}(1) &= 1 - (1-1)^{1/m} = 1 - 0 = 1
\end{align*}
Thus, Condition 1 is satisfied.

\noindent \textbf{Condition 2: Strict monotonicity.} \\
For any $u \in (0,1)$ and $m > 1$, the term $(1-u)$ is strictly positive. Consequently, the exponential evaluation $(1-u)^{\frac{1-m}{m}}$ evaluates to a real, strictly positive number. Because $1/m > 0$, it follows from Equation \eqref{eq:deriv1} that $t'_{cr}(u) > 0 \ \forall u \in (0,1)$. The function is clearly continuous and differentiable on the open interval, satisfying Condition 2.

\noindent \textbf{Condition 3: Strict convexity.} \\
For any $u \in (0,1)$ and $m > 1$, the term $(1-u)$ is strictly positive. Furthermore, because $m > 1$, the coefficient $(m-1)/m^2 > 0$. Therefore, from Equation \eqref{eq:deriv2}, the product of strictly positive terms yields $t''_{cr}(u) > 0 \ \forall u \in (0,1)$, proving strict convexity and satisfying Condition 3.

\noindent \textbf{Condition 4: Origin sparsity.} \\
Evaluating the right-hand limit of the first derivative as $u \to 0^+$:
\begin{equation*}
    \lim_{u \to 0^+} t'_{cr}(u) = \lim_{u \to 0^+} \frac{1}{m}(1-u)^{\frac{1-m}{m}} = \frac{1}{m} (1)^{\frac{1-m}{m}} = \frac{1}{m}
\end{equation*}
Since $m$ is finite and strictly greater than $1$, the limit evaluates to a finite, strictly positive constant $1/m > 0$. Condition 4 is thus satisfied.

\noindent \textbf{Condition 5: Asymptotic stability threshold.} \\
We evaluate the spectral stability threshold inverse $1/\kappa(t_\text{cr}, u)$ at the uniform state $u = 1/C$:
\begin{equation*}
    \frac{1}{\kappa(t_\text{cr}, 1/C)} = \frac{t_\text{cr}(1/C) \cdot t''_{cr}(1/C)}{2[t'_{cr}(1/C)]^2}
\end{equation*}
Substituting the definitions of the function and its derivatives, and simplifying:
\begin{align*}
    \frac{1}{\kappa(t_\text{cr}, 1/C)} &= \frac{\left[1 - \left(1 - \frac{1}{C}\right)^{1/m}\right] \left[ \frac{m-1}{m^2}\left(1 - \frac{1}{C}\right)^{\frac{1-2m}{m}} \right]}{2 \left[ \frac{1}{m}\left(1 - \frac{1}{C}\right)^{\frac{1-m}{m}} \right]^2} \\
    &= \frac{m-1}{2} \frac{1 - \left(1 - \frac{1}{C}\right)^{1/m}}{\left(1 - \frac{1}{C}\right)^{1/m}} \\
    &= \frac{m-1}{2} \left[ \left(1 - \frac{1}{C}\right)^{-1/m} - 1 \right] \\
    &= \frac{m-1}{2} \left[ \left(\frac{C}{C-1}\right)^{1/m} - 1 \right]
\end{align*}
To evaluate the inner limit as $m \to \infty$, we utilise the standard limit property $\lim\limits_{x \to 0} (a^x - 1)/x = \log(a)$. Letting $x = 1/m$, we obtain:
\begin{equation*}
    \lim_{m \to \infty} \frac{1}{\kappa(t_\text{cr}, 1/C)} = \frac{1}{2} \log\left(\frac{C}{C-1}\right)
\end{equation*}
Finally, taking the outer limit as the number of clusters $C \to \infty$:
\begin{equation*}
    \lim_{C \to \infty} \frac{1}{2} \log\left(\frac{C}{C-1}\right) = \frac{1}{2} \log(1) = 0
\end{equation*}
The spectral threshold converges strictly to zero, satisfying Condition 5.

% \noindent \textbf{Condition 6: Asymptotic structural preservation.} \\
% We construct the marginal derivative ratio $R_m(u_{ik}, u_{il})$:
% \begin{align*}
%     R_m(u_{ik}, u_{il}) &= \frac{t'_{cr}(u_{ik})}{t'_{cr}(u_{il})} = \frac{\frac{1}{m}(1-u_{ik})^{\frac{1-m}{m}}}{\frac{1}{m}(1-u_{il})^{\frac{1-m}{m}}} \\
%     &= \left( \frac{1-u_{ik}}{1-u_{il}} \right)^{\frac{1}{m} - 1}
% \end{align*}
% Taking the pointwise limit as $m \to \infty$, the exponent $(1/m - 1) \to -1$. Thus:
% \begin{equation*}
%     R_\infty(u_{ik}, u_{il}) = \lim_{m \to \infty} R_m = \left( \frac{1-u_{ik}}{1-u_{il}} \right)^{-1} = \frac{1-u_{il}}{1-u_{ik}}
% \end{equation*}
% This limit function is strictly continuous for all $u_{ik}, u_{il} \in (0,1)$. By fixing $u_{il}$ and allowing $u_{ik} \to 1^-$, $R_\infty \to \infty$. By fixing $u_{ik}$ and allowing $u_{il} \to 1^-$, $R_\infty \to 0^+$. Thus, the image of $R_\infty$ spans $(0, \infty)$, preventing asymptotic squashing and satisfying Condition 6.
\end{proof}

\color{black}
\section{Connection of anisotropy constraint and $W_\text{eff}$}\label{appen5}
\setcounter{figure}{0}

\begin{theorem}\label{thm:weff_alpha}
    Let $\boldsymbol\theta \in \mathbb R^p_{>0}$ be a vector of bandwidths with $p \ge 2$ and $\max_j\theta_j/\min_j\theta_j \le \alpha$, where $\alpha \ge 1$. Define $W_j = (1/\theta_j)/\sum_k(1/\theta_k)$ and $W_{\text{eff}} = 1/\sum_j W_j^2$. Then, 
    $$W_{\text{eff}} \ge \min_{0 \le k \le p}\frac{(k\alpha + p - k)^2}{k\alpha^2 + p - k},$$
    with equality when $k$ bandwidths equal $\min_j\theta_j$ and the other $p - k$ equal $\max_j\theta_j$, for the minimising $k$.
\end{theorem}

\begin{proof}
    Define $w_j = 1/\theta_j > 0$ and $\mathbf{w} = (w_1, \ldots, w_p)^\top \in \mathbb{R}^p_{>0}$. Then, $W_j = w_j/\sum_{k=1}^p w_k$ and thus, a straightforward calculation gives
    \begin{equation*}
        W_{\text{eff}} = \frac{\left( \sum\limits_{k=1}^p w_k\right)^2}{\sum\limits_{j=1}^p w_j^2} = \frac{\left( \mathbf{1}_p^\top \mathbf{w} \right)^2}{\| \mathbf{w} \|^2}.
    \end{equation*}
    It is also easy to see that $W_\text{eff}$ is unaffected by rescaling, that is, for $C > 0$, replacing $\mathbf{w}$ by $C\mathbf{w}$ does not change the derived expression for $W_\text{eff}$. Let $\tilde{w}_j = w_j / \min_kw_k$, so $W_\text{eff} = (\mathbf{1}_p^\top \tilde{\mathbf{w}})^2/\| \tilde{\mathbf{w}}\|^2$. Since $w_j = 1/\theta_j$, $\max_j w_j = 1/\min_j \theta_j$ and $\min_j w_j = 1/\max_j \theta_j$. Therefore,
    \begin{equation*}
        \frac{\max\limits_j w_j}{\min\limits_j w_j} = \frac{\max\limits_j \theta_j}{\min\limits_j \theta_j} \le \alpha,
    \end{equation*}
    which gives $1 \le \tilde{w}_j \le \alpha$ and equivalently, $\tilde{\mathbf{w}} \in [1, \alpha]^p$. In order to obtain a lower bound for $W_\text{eff}$, we seek to find $\mathbf{v} \in [1, \alpha]^p$ such that $f(\mathbf{v}) = (\mathbf{1}_p^\top \mathbf{v})^2/\| \mathbf{v} \|^2$ is minimised. We know that a minimum does exist since $f$ is continuous and we are working in the hypercube $[1, \alpha]^p$ which is closed and bounded. We fix all entries of $\mathbf{v}$ except for component $j$ and define $\sigma_1 = \sum_{i \ne j} v_i \ge p-1$ and $\sigma_2 = \sum_{i \ne j} v_i^2 \ge p-1$. Then,
    \begin{align*}
        f(\mathbf{v}) & = \frac{\left( \sigma_1 + v_j \right)^2}{\sigma_2 + v_j^2}\\
        \implies \frac{\mathrm{d} f}{\mathrm{d} v_j} & = \frac{2\left(\sigma_1 + v_j\right)\left(\sigma_2 - \sigma_1 v_j \right)}{\left( \sigma_2 + v_j^2 \right)^2}.
    \end{align*}
    The derivative has the sign of $\sigma_2 - \sigma_1v_j$ as both the denominator and $\sigma_1 + v_j$ are positive terms, so $f$ increases when $v_j < \sigma_2 / \sigma_1$ and decreases when $v_j > \sigma_2 / \sigma_1$. Hence, for fixed values of the remaining components, $f$ is minimised at $v_j = 1$ or $v_j = \alpha$. Starting from a minimiser of $f$ and replacing each component in turn by whichever of $1$ and $\alpha$ gives the smaller value of $f$ never leads to an increase of $f$, so its minimum over $[1, \alpha]^p$ is attained at a corner of the hypercube. A corner is described by $k$ components being equal to $\alpha$ and the remaining $p-k$ being equal to $1$, where $0 \le k \le p$. For such a corner $\mathbf{v}$, $\mathbf{1}_p^\top \mathbf{v} = k\alpha + p-k$ and $\| \mathbf{v} \|^2 = k\alpha^2 + p-k$, which finally gives
    \begin{equation*}
        W_\text{eff} = f(\tilde{\mathbf{w}}) \ge \min\limits_{\mathbf{v} \in [1, \alpha]^p} f(\mathbf{v}) = \min\limits_{0 \le k \le p} \frac{(k\alpha + p - k)^2}{k\alpha^2 + p - k}.
    \end{equation*}
    We conclude by showing that the bound is attained. Let $k^*$ be the minimising $k$ and set $\theta_j = \theta_0$ for $k^*$ of the variables and $\theta_j = \alpha\theta_0$ for the rest. These bandwidths satisfy the anisotropy constraint and give $\tilde{\mathbf w}$ equal to the corner with $k^*$ entries equal to $\alpha$, so $W_{\text{eff}}$ equals the bound.
\end{proof}

\begin{figure}[!ht]
    \centering
    \includegraphics[width=\linewidth]{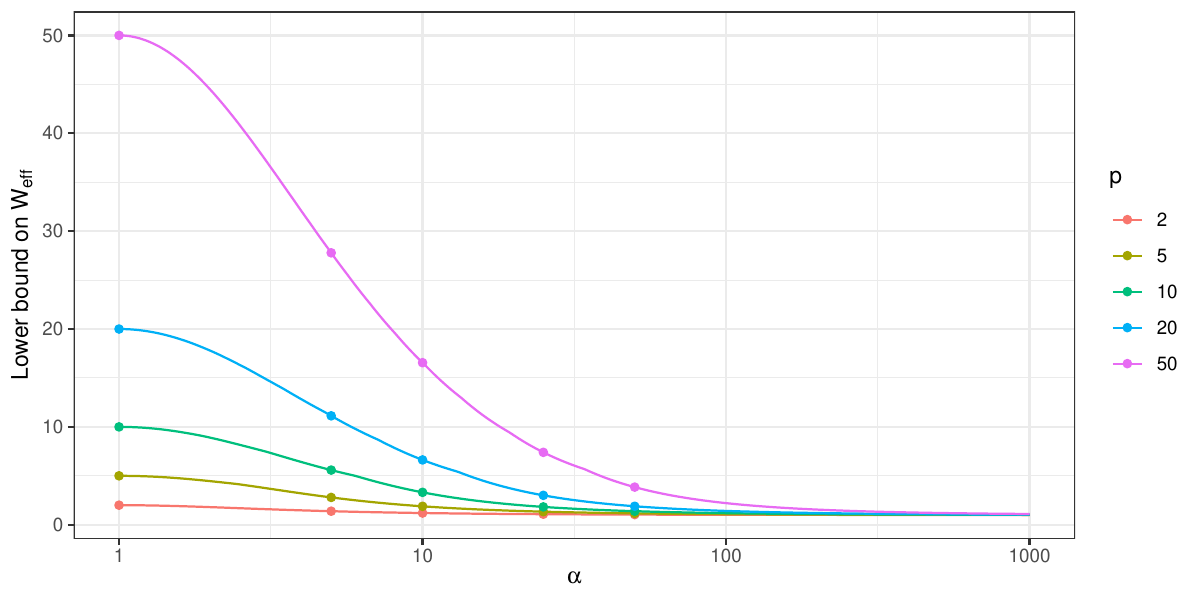}
    \caption{Lower bound on $W_\mathrm{eff}$ as a function of $\alpha$ (log scale) for different numbers of variables $p$. Points mark $\alpha \in \{1, 5, 10, 25, 50\}$.}
    \label{fig:weff_bounds}
\end{figure}

\color{black}

\bibliographystyle{elsarticle-num-names} 
\bibliography{export}

\end{document}